\documentclass[final]{siamltex}

\usepackage{amsmath}
\usepackage{amsfonts}
\usepackage{latexsym}
\usepackage{epsfig}
\usepackage{amssymb}
\usepackage{epstopdf}

\newcommand{\tra}[1]{\,{\vphantom{#1}}^{\textsc{t}}\!{#1}}

\newtheorem{remark}{{\it Remark}}[section]

\begin{document}

\title{Discrete calculus of variation for homographic configurations in celestial mechanics}

\author{Philippe Ryckelynck\footnotemark[1]~~ Laurent Smoch\footnotemark[1]} 

\thanks{{ULCO, LMPA, F-62100 Calais, France.\hspace{1cm}e-mail: \{ryckelyn,smoch\}@lmpa.univ-littoral.fr\\
\indent \indent Univ Lille Nord de France, F-59000 Lille, France. CNRS, FR 2956, France.}}

\maketitle
\baselineskip15pt

\begin{abstract}
We provide in this paper the discrete equations of motion for the newtonian $n$-body problem deduced from the quantum calculus of variations (Q.C.V.) developed in \cite{Cre,CFT,RS1,RS2}. These equations are brought into the usual lagrangian and hamiltonian formulations of the dynamics and yield sampled functional equations involving generalized scale derivatives. We investigate especially homographic solutions to these equations that we obtain by solving algebraic systems of equations similar to the classical ones. When the potential forces are homogeneous, homographic solutions to the discrete and classical equations may be related through an explicit expansion factor that we provide. Consequently, perturbative equations both in lagrangian and hamiltonian formalisms are deduced.
\end{abstract}

\begin{keywords}
$n$-body problem, homographic solutions, central configurations, periodic solutions, discretization
\end{keywords}

\begin{AMS}
49J21, 70F10, 49J15, 70F07
\end{AMS}

\section{Introduction}

This paper is devoted to the application of discrete calculus of variations
(see Cresson \cite{Cre}, Torres and Frederico \cite{CFT}, Ryckelynck and
Smoch \cite{RS1,RS2}) to celestial mechanics. We focus on central configurations and especially on
libration points or regular polygonal solutions \cite{El1,El2}, which are the most well-known periodic solutions. These motions are entirely explicit
provided we can solve specific algebraic equations for the coordinates in a
rotating frame. Another family of periodic solutions consists in choreographic solutions \cite{CM} which are more involved and obtained through
topological arguments. In this respect we have obtained in \cite{RS2}
choreographic solutions to quadratic lagrangian systems either in classical
and discrete contexts. Both families of solutions are infinite and give rise
to a huge number of theoretical and numerical works. 

The discrete calculus of variations deals with sets of non-differentiable curves by substituting the classical derivative for a so-called generalized scale derivative.  
Formally we used in \cite{RS1,RS2} the following discretization operator 
\begin{equation}
\Box \mathbf{u}(t)=\sum_{\ell =-N}^N\frac{\gamma _\ell }\varepsilon \chi (t+\ell \varepsilon )\mathbf{u}(t+\ell \varepsilon ),  
\label{ourbox}
\end{equation}
for all $t\in [t_0,t_f]$, for all $\mathbf{u}:[t_0,t_f]\rightarrow \mathbb{R}%
^d$, $d\in \mathbb{N}^{\star }$, and where $\chi $ denotes the
characteristic function of $[t_0,t_f]$. The notation $\Box ^{\star }$ shall
denote the adjoint operator to $\Box $ which is obtained from $\Box $ by
reversing its coefficients $\gamma _\ell $. The investigation of the convergence of $\Box \mathbf{u}$ and $-\Box ^{\star }\mathbf{u}$ to $%
\dot{\mathbf{u}}$ has been undertaken in \cite[Proposition 2.3]{RS1}. We have already pointed out the fact that for all $\mathbf{u}\in \mathcal{C}^2([t_0,t_f],\mathbb{R}^d)$ and for all $t\in ]t_0,t_f[$, 
\begin{equation}
\lim_{\varepsilon \to 0}\Box \mathbf{u}(t)=\lim_{\varepsilon \to 0}-\Box^{\star }\mathbf{u}(t)=\dot{\mathbf{u}}(t)
\label{convbox}
\end{equation} 
locally uniformly in $]t_0,t_f[$ if and only if
\begin{equation}
\sum_\ell \gamma _\ell =0\mbox{ and }\sum_\ell \ell \gamma _\ell =1.
\label{convcond}
\end{equation} 

These notations and properties being introduced, we consider a system of $n$ particles $P_i$, with mass $m_i$, located at points $\mathbf{x}_i=(x_{ik})_k\in \mathbb{R}^d$ where $i=1,\ldots ,n$ and $k=1,\ldots ,d$. The distance $r_{ij}$ between $P_i$ and $P_j$ is defined by $\displaystyle r_{ij}^2=%
\sum_{k=1}^d(x_{ik}-x_{jk})^2$. We assume that there exist $\frac{n(n-1)}2$
functions of forces $f_{ij}(r_{ij})$ determining the interactions between
each pair of particles $(P_i,P_j)$. So we set for all $\mathbf{x},\mathbf{y}\in(\mathbb{R}^d)^n$
\begin{equation}
\displaystyle T(\mathbf{y})=\frac 12\sum_{i=1}^n\sum_{k=1}^dm_iy_{ik}^2,~
\displaystyle U(\mathbf{x})=\sum_{i<j}f_{ij}(r_{ij}).
\label{TU}
\end{equation}
Accordingly, the Lagrangian and the Hamiltonian are defined for all configurations of particles $\mathbf{x}=(\mathbf{x}_1,\ldots ,\mathbf{x}_n)\in \mathcal{C}^0([t_0,t_f],(\mathbb{R}^d)^n)$ in the classical and discrete settings respectively by
\begin{eqnarray}
\mathcal{L}_c =T(\dot {\mathbf{x}})+U(\mathbf{x}),~\mathcal{H}_c=T(\dot {%
\mathbf{x}})-U(\mathbf{x}), \label{LUTH1}\\
\mathcal{L}_d =T(\Box \mathbf{x})+U(\mathbf{x}),~\mathcal{H}_d=T(\Box 
\mathbf{x})-U(\mathbf{x}). \label{LUTH2}
\end{eqnarray}
The homogeneous potential functions of the shape $f_{ij}(r)=\mu _{ij}r^\beta$, with some common exponent $\beta\in\mathbb{Q}-\{0,2\}$, constitute a case of
special interest. We are particularly interested in the gravific interaction, described by $\beta =-1$ and $\mu
_{ij}=gm_im_j$, so that $f_{ij}(r)=g\frac{m_im_j}r$. 

Now, when working in a rotating frame with constant pulsation $\omega$, we look for the homographic solutions to the $n$-body problem, i.e. the solutions to the equations of motion of the shape 
\begin{equation}
x_{i1}(t)=a_i(t)\cos \omega t-b_i(t)\sin \omega t,~x_{i2}(t)=a_i(t)\sin\omega t+b_i(t)\cos \omega t,  
\label{homographicsol}
\end{equation}
for some functions $a_i(t)$ and $b_i(t)$, according to the additional conditions $x_{ij}=0\mbox{ for }%
j\geq 3$. The connection between homographic and central
configurations is explored in \cite{BP,SS}. 

The paper is organized as follows. In Section 2, we give the four systems of equations of motion for the $n$-body problem, either in the lagrangian or hamiltonian formulations and either classical or discrete settings. This being done, we discuss the existence of constants of motion and galilean equilibria. In Section 3, 
we introduce the additionnal operators $V_c,V_s,W_c,W_s$, used when expressing the Euler-Lagrange and hamiltonian equations in a rotating frame to be further developed in Section 4. There we provide the convenient formulas for $\mathcal{L}_c,\mathcal{L}_d,\mathcal{H}_c,\mathcal{H}_d$, and the four corresponding sets of
equations of motion. In Section 5, we determine relative equilibria solutions to the $n$-body
problem that is to say, the solutions (\ref{homographicsol}) obtained such that the functions $a_i(t)$ and $b_i(t)$ are constant w.r.t. time.
When the potential functions are homogeneous, we show that the solutions to the discrete Euler-Lagrange equations are homothetic to those to the classical Euler-Lagrange equations. The homothety ratio $\varphi (\varepsilon )$ that we call the expansion factor, is determined and its convergence as $\varepsilon$ tends to $0$ is studied. Finally, in Section 6, we provide some numerical experiments and results which illustrate our analysis.

\section{Discrete and classical equations of motion for the $n$-body problem in a galilean frame}

\subsection{Euler-Lagrange equations}

To motivate our work, let us recall the well-known classical equations of motion for the newtonian $n$-body problem according to (\ref{LUTH1})
\begin{equation}
\displaystyle m_i\ddot x_{ik}=\sum_{j\neq i}{f}_{ij}'(r_{ij})\frac{%
x_{ik}-x_{jk}}{r_{ij}},
\label{celuniv}
\end{equation}
where $i\in\{1,\ldots,n\},k\in\{1,\ldots,d\}$. By considering the discrete Euler-Lagrange equations introduced in \cite{RS1}, we may deduce the discrete analogous equations to (\ref{celuniv}).
\begin{proposition}
Let a system of $n$ particles interacting according to (\ref{LUTH2}) then the discrete equations of motion are, for all $i$ and $k$,
\begin{equation}
\displaystyle -m_i\Box ^{\star }\Box x_{ik}=\sum_{j\neq i}{f}_{ij}'(r_{ij})\frac{x_{ik}-x_{jk}}{r_{ij}}.  
\label{mcelter}
\end{equation}
\end{proposition}
\begin{proof} We deduced in \cite[Theorem 4.1]{RS1} the discrete Euler-Lagrange
equations which may be written as 
\begin{equation}
\Box^\star\frac{\partial \mathcal{L}_d}{\partial \Box x_{ik}}(t,\mathbf{x}%
(t),\Box \mathbf{x}(t))+\frac{\partial \mathcal{L}_d}{\partial x_{ik}}(t,%
\mathbf{x}(t),\Box \mathbf{x}(t))=0,  
\label{intrinsic}
\end{equation}
for all $i\in\{1,\ldots,n\}$ and $k\in\{1,\ldots,d\}$. Equations (\ref{mcelter}) are an easy consequence of (\ref{intrinsic}) applied to $\mathcal{L}_d$ given in (\ref{LUTH2}).
\end{proof}

\begin{remark}\rm
The explicit value of the left-hand side of (\ref{mcelter}) is given help to the following formula 
\begin{equation}
\Box ^{\star }\Box \mathbf{f}(t)=\frac 1{\varepsilon ^2}\sum_{{\tiny 
\begin{array}[t]{c}
|\ell|\leq 2N \\ 
|j|\leq N \\ 
|\ell+j|\leq N
\end{array}
}}\hspace{-0.2cm}\gamma _{\ell +j}\gamma _j\chi (t-j\varepsilon )\chi
(t+\ell \varepsilon )\mathbf{f}(t+\ell \varepsilon ).  
\label{mchomat2}
\end{equation}
which is excerpt from \cite{RS1} and shall be used throughout the paper. In that sense, formula (\ref{mcelter}) may be thought as a system of functional
delayed equations.
\end{remark}

\subsection{Hamilton's equations}

In order to provide the hamiltonian equations equivalent to the classical and discrete Euler-Lagrange equations, 
we introduce the components of the momenta $p_{ik}=\frac{\partial \mathcal{L}_c}{\partial\dot {x_{ik}}}=\frac{\partial T}{\partial \dot {x_{ik}}}$. We suppose that
the hessian matrix of $T$ is definite. Then the mapping defined by $\mathbb{R}^{2dn}\rightarrow \mathbb{R}^{2dn}$, $(\mathbf{x}_i,\dot {\mathbf{x}}_i)_{i\in \{1,\ldots ,n\}}\mapsto (\mathbf{x}_i,\mathbf{p}_i)_{i\in\{1,\ldots ,n\}}$ is locally one to one. So we may express $T(\dot {\mathbf{x}})$, $U(\mathbf{x})$ as functions of $\mathbf{x}$ and $\mathbf{p}=(\mathbf{p}_1,\ldots ,\mathbf{p}_n)\in \mathcal{C}^0([t_0,t_f],(\mathbb{R}^d)^n)$ to
obtain $T-U=\mathcal{H}_c(\mathbf{x},\mathbf{p})$. Similarly, in the discrete setting, the convenient coordinates of the momenta are $p_{ik}=\frac{\partial \mathcal{L}_d}{\partial \Box {x}_{ik}}$ and we obtain accordingly the Hamiltonian $\mathcal{H}_d(\mathbf{x},\mathbf{p})$. Due to (\ref{LUTH1}) and (\ref{LUTH2}), we note that the two hamiltonian functions are formally the same function that we denote $\mathcal{H}(\mathbf{x},%
\mathbf{p})$. As a well known result, the equations (\ref{celuniv}) are equivalent to the systems
\begin{equation}
\dot p_{ik}=-\frac{\partial \mathcal{H}}{\partial x_{ik}},~\dot x_{ik}=\frac{\partial \mathcal{H}}{\partial p_{ik}},  
\label{eqHamcont}
\end{equation}
We may easily state an analogue of (\ref{eqHamcont}) in the discrete setting.
\begin{proposition}
The equations (\ref{mcelter}) are equivalent to the systems of Hamilton's equations 
\begin{equation}
\Box ^{\star }{p}_{ik}=\frac{\partial \mathcal{H}}{\partial x_{ik}},~\Box {x}_{ik}=\frac{\partial \mathcal{H}}{\partial p_{ik}}.  
\label{eqHamdisc}
\end{equation}
\end{proposition}
\begin{proof}
In this fairly simple setting we have $p_{ik}=\frac{\partial \mathcal{L}_d}{\partial \Box{x}_{ik}}%
=m_i\Box x_{ik}$. Thus, equation (\ref{intrinsic}) may be rewritten as $\Box^\star p_{ik}=-\frac{\partial \mathcal{L}_d}{\partial
x_{ik}}=-\frac{\partial U}{\partial x_{ik}}=\frac{\partial \mathcal{H}}{%
\partial x_{ik}}$. Next, since $T=\frac12\sum_{i=1}^n\sum_{k=1}^d\frac{1}{m_i}p_{ik}^2$, it is obvious that $\frac{\partial \mathcal{H}}{\partial p_{ik}}=\frac{\partial T}{\partial p_{ik}}=\frac1{m_i}p_{ik}=\Box x_{ik}$. Thus, we have proved that (\ref{mcelter}) implies (\ref{eqHamdisc}) and the
converse is easy.
\end{proof}

We shall see in Remark \ref{rem41} that the Hamilton's equations are not covariant and highlight the restrictive assumptions (\ref{LUTH1}) and (\ref{LUTH2}). The computation of integrals of motion is done in Subsection \ref{sec53} in the particular case of central configurations.

\section{The four functional operators $V_c,V_s,W_c,W_s$.}

In this section we introduce four continuous linear operators between the euclidean function space $\mathcal{C}^0([t_0,t_f],\mathbb{R}^d)$ and the function space of which the elements are the piecewise continuous functions vanishing outside $[t_0-N\varepsilon,t_f+N\varepsilon]$. The first space is equipped with the ordinary scalar product, i.e. $\displaystyle \langle\mathbf{f},\mathbf{g}\rangle_0=\int_{t_0}^{t_f}\langle\mathbf{f}(t),\mathbf{g}(t)\rangle dt$, while the second one is endowed with the same scalar product except for the domain of integration which is $[t_0-N\varepsilon,t_f+N\varepsilon]$. As usual, we use the notation $\star$ for denoting the adjoint of an operator.

\begin{proposition}
There exist four uniquely well-defined operators $V_c,V_s,W_c,W_s$ such that
\begin{eqnarray}
\Box (\mathbf{f}(t)\cos (\omega t)) & = & V_c(\mathbf{f})(t)\cos (\omega t)-V_s(\mathbf{f})(t)\sin (\omega t),\label{op1}\\
\Box (\mathbf{f}(t)\sin (\omega t)) & = &V_s(\mathbf{f})(t)\cos (\omega t)+V_c(\mathbf{f})(t)\sin (\omega t),\label{op2}\\
\Box ^{\star }\Box (\mathbf{f}(t)\cos (\omega t)) & = & W_c(\mathbf{f})(t)\cos (\omega t)-W_s(\mathbf{f})(t)\sin (\omega t),\label{op3}\\
\Box ^{\star }\Box (\mathbf{f}(t)\sin (\omega t)) & = & W_s(\mathbf{f})(t)\cos (\omega
t)+W_c(\mathbf{f})(t)\sin (\omega t).\label{op4}
\end{eqnarray}
for all mappings $\mathbf{f}:[t_0,t_f]\rightarrow\mathbb{R}^d$. The four operators $V_c,V_s,W_c,W_s$ are connected through the formulas
\begin{equation}
W_c=V_c^{\star }V_c+V_s^{\star }V_s=W_c^\star,~~~~W_s=V_c^{\star }V_s-V_s^{\star }V_c=-W_s^\star.
\label{op5}
\end{equation}
Lastly, let $(\gamma _i)\in \mathbb{C}^{2N+1}$ be such that $\sum_\ell \gamma _\ell=0$ and 
$\sum_\ell \ell \gamma _\ell =1$. Then, for all $\mathbf{f}\in \mathcal{C}^2([t_0,t_f],\mathbb{R}^d)$,
\begin{equation}
\displaystyle \lim_{\varepsilon \to 0}V_c(\mathbf{f})=\dot{\mathbf{f}} \mbox{ and }
\displaystyle \lim_{\varepsilon \to 0}V_s(\mathbf{f})=\omega\mathbf{f},
\label{convV}
\end{equation}
\begin{equation}
\displaystyle \lim_{\varepsilon \to 0}W_c(\mathbf{f})=\omega ^2\mathbf{f}-\ddot{\mathbf{f}}\mbox{ and }
\displaystyle \lim_{\varepsilon \to 0}W_s(\mathbf{f})=-2\omega \dot{\mathbf{f}}
\label{convW}
\end{equation}
locally uniformly in $]t_0,t_f[$.
\label{prop31}
\end{proposition}

\begin{proof}
Let us introduce the four operators
\begin{center}
$\displaystyle V_c(\mathbf{f}) =\frac{1}{\varepsilon}\sum_{|j|\leq N}\gamma _j \cos (\omega j\varepsilon )\chi (t+j\varepsilon )\mathbf{f}(t+j\varepsilon )$,\\
$\displaystyle V_s(\mathbf{f}) =\frac{1}{\varepsilon}\sum_{|j|\leq N}\gamma _j \sin (\omega
j\varepsilon )\chi (t+j\varepsilon )\mathbf{f}(t+j\varepsilon )$,
\end{center}
\begin{center}
$\displaystyle W_c(\mathbf{f})=\frac 1{\varepsilon ^2}\sum_{\tiny
\begin{array}[t]{c}
|\ell|\leq 2N \\ 
|j|\leq N \\ 
|\ell +j|\leq N
\end{array}
}\hspace{-0.2cm}\gamma _{\ell +j}\gamma _j\chi (t-j\varepsilon )\chi (t+\ell
\varepsilon )\cos (\omega\ell \varepsilon)\mathbf{f}(t+\ell \varepsilon )$,
$\displaystyle W_s(\mathbf{f})=\frac 1{\varepsilon ^2}\sum_{\tiny
\begin{array}[t]{c}
|\ell|\leq 2N \\ 
|j|\leq N \\ 
|\ell +j|\leq N
\end{array}
}\hspace{-0.2cm}\gamma _{\ell +j}\gamma _j\chi (t-j\varepsilon )\chi (t+\ell
\varepsilon )\sin (\omega\ell \varepsilon)\mathbf{f}(t+\ell \varepsilon )$.
\end{center}
Straightforward computations show that equations (\ref{op1}) to (\ref{op4}) hold. Let us recall from \cite[Lemma 2.1]{RS2} that the adjoint of an operator $\Box$ defined by (\ref{ourbox}) is obtained by reversing its coefficients $(\gamma_j)_{j}$ to $(\gamma_{-j})_j$. Substituting $\Box^\star$ for $\Box$ in (\ref{op1}) and (\ref{op2}) implies two new operators $\tilde{V}_c$ and $\tilde{V}_s$, whose coefficients are the sequences $(\gamma_{-j}\cos(\omega j\varepsilon))_j$ and $(\gamma_{-j}\sin(\omega j\varepsilon))_j$ respectively, that is to say the reversed sequences of coefficients of the operators $V_c^\star$ and $-V_s^\star$ respectively. It follows that the two formulas $W_c=V_c^{\star }V_c+V_s^{\star }V_s$ and $W_s=V_c^{\star }V_s-V_s^{\star }V_c$ hold, as consequences of (\ref{op1}) and (\ref{op2}) when using the operators $\Box$ and $\Box^\star$ and mentioning the unicity of coefficients in (\ref{op3}) and (\ref{op4}). Let us remark that formulas (\ref{op5}) imply that $W_c$ is symmetric and $W_s$ is skew-symmetric.

An inspection of the proof given in \cite[Proposition 2.3]{RS1} shows that the result of convergence (\ref{convbox}) extends to $\mathcal{C}^2$-piecewise functions $\mathbf{u}$. This being observed, we may deduce the two last results of the property. Since the proofs are similar, we focus especially on (\ref{convW}). Equations (\ref{op3}) and (\ref{op4}) may be rewritten as
\[
\begin{pmatrix}
\cos(\omega t) & -\sin(\omega t)\\
\sin(\omega t) & \cos(\omega t)
\end{pmatrix}\begin{pmatrix}
W_c(\mathbf{f})\\
W_s(\mathbf{f})
\end{pmatrix}=%
\begin{pmatrix}
\Box^\star\Box (\mathbf{f}\cos(\omega t))\\
\Box^\star\Box (\mathbf{f}\sin(\omega t))
\end{pmatrix}
\]
and, as a consequence, we get the identity
\begin{equation}
{\small %
\begin{pmatrix}
\cos(\omega t) & -\sin(\omega t)\\
\sin(\omega t) & \cos(\omega t)
\end{pmatrix}%
\begin{pmatrix}
(\ddot{\mathbf{f}}-\omega^2\mathbf{f})+W_c(\mathbf{f})\\
2\omega \dot{\mathbf{f}}+W_s(\mathbf{f})
\end{pmatrix}=%
\begin{pmatrix}
(\mathbf{f}\cos(\omega t))''+\Box^\star\Box (\mathbf{f}\cos(\omega t))\\
(\mathbf{f}\sin(\omega t))''+\Box^\star\Box (\mathbf{f}\sin(\omega t))
\end{pmatrix}. }
\label{convsys}
\end{equation}
The matrix of this system being invertible, the r.h.s of (\ref{convsys}) tends to 0 as $\varepsilon$ tends to $0$ if and only if its l.h.s tends to 0, locally uniformly  on each interval of the shape $[t_0+\delta ,t_f-\delta ]$, for all $\delta >0$. By using the formula (\ref{ourbox}) in the interval $[t_0+N\varepsilon ,t_f-N\varepsilon ]$, we see
that $\Box \mathbf{f}(t)=\frac{\gamma _{-N}}\varepsilon \mathbf{f}(t-N\varepsilon )+\ldots+\frac{\gamma _N}\varepsilon \mathbf{f}(t+N\varepsilon )$. Since $\mathbf{f}$ is $C^2$ in $[t_0,t_f]$, $\Box \mathbf{f}$ is $\mathcal{C}^2$ in $[t_0+N\varepsilon,t_f-N\varepsilon ]$ and we may apply \cite[Proposition 2.3]{RS1} to state first that $-\Box ^{\star }\Box \mathbf{f}$ tends to $\ddot{\mathbf{f}}$ as $\varepsilon $ tends to 0 and next, help to (\ref{convsys}), that (\ref{convV}) and (\ref{convW}) hold.
\end{proof}

\begin{remark}\rm
We may interprete the four identities (\ref{op1}) to (\ref{op4}) as specialized Leibniz formulas of order 1 and 2. In \cite[Theorem 3.1]{RS1}, we already expressed the remainder $\Box(\mathbf{f}\mathbf{g})-\mathbf{f}\Box\mathbf{g}-\mathbf{g}\Box\mathbf{f}$ in a generalized Leibniz formula for a specific class of operators $\Box$.
\end{remark}

\begin{remark}\rm
We notice that the operators $V_c$ and $V_s$, whose coefficients are the sequences $(\gamma_j\cos(\omega j\varepsilon))_j$ and $(\gamma_j\sin(\omega j\varepsilon))_j$ respectively, are of the shape (\ref{ourbox}) if we do not consider the fact that these coefficients depend on $\varepsilon$. In contrast, this is not the case for $W_c$ and $W_s$. 
\end{remark}

\section{Equations of motion in a rotating frame}

The aim of this section is to provide, when it is possible, the classical and discrete equations of motion in a rotating frame by using the lagrangian and hamiltonian formalisms and the Legendre transform. From now on, we drop $t$ from the following dynamic variables since it is clear. 

\subsection{Euler-Lagrange equations in the rotating frame}

To begin with, we shall suppose that the cartesian coordinates $(x_{ik}(t))_{i,k}$ in the classical and discrete settings are of the shape (\ref{homographicsol}), expressed respectively through the $2\times(2n)$ functions $a_i(t),b_i(t)$ and $A_i(t,\varepsilon),B_i(t,\varepsilon)$, which may be thought as perturbative variables around the relative equilibria. Therefore, the distances between the points lying in the plane $(x_{i1},x_{i2})$, in the classical and discrete settings, are respectively  equal to
\begin{center}
$\displaystyle r_{ij}^2=(a_i-a_j)^2+(b_i-b_j)^2$ and $R_{ij}^2=(A_i-A_j)^2+(B_i-B_j)^2$.
\end{center} 

Let us provide now the perturbative Euler-Lagrange equations in the rotating frame. In the classical setting, we use (\ref{homographicsol}) to compute $\dot x_{i1},\dot
x_{i2},\ddot x_{i1},\ddot x_{i2}$ and we plug these functions in (\ref{celuniv}). Help to suitable linear combinations, we obtain 
\begin{equation}
\begin{array}{l}
\displaystyle m_i(\ddot a_i-2\omega \dot b_i-\omega ^2a_i)=\sum_{j\neq i}\frac{%
{f}_{ij}'(r_{ij})}{r_{ij}}(a_i-a_j),\\
\displaystyle m_i(\ddot b_i+2\omega \dot a_i-\omega ^2b_i)=\sum_{j\neq i}\frac{%
{f}_{ij}'(r_{ij})}{r_{ij}}(b_i-b_j).
\end{array}
\label{CELperturb}
\end{equation}
In the discrete setting, we use (\ref{mcelter}), (\ref{op3}) and (\ref{op4}) to obtain 
\begin{center}
$\Box^\star\Box x_{i1}(t)=(W_c(A_i)-W_s(B_i))(t)\cos (\omega t)-(W_s(A_i)+W_c(B_i))(t)\sin (\omega t)$,\\
$\Box^\star\Box x_{i2}(t)=(W_s(A_i)+W_c(B_i))(t)\cos (\omega t)+(W_c(A_i)-W_s(B_i))(t)\sin (\omega t)$.
\end{center}
Help to the same linear combinations than in the classical case, we obtain the following equations of motion 
\begin{equation}
\begin{array}{l}
\displaystyle -m_i(W_c(A_i)-W_s(B_i))=\sum_{j\neq i}\frac{{f}_{ij}'(R_{ij})}{R_{ij}}(A_i-A_j),\\
\displaystyle -m_i(W_s(A_i)+W_c(B_i))=\sum_{j\neq i}\frac{{f}_{ij}'(R_{ij})}{R_{ij}}(B_i-B_j).  
\end{array}
\label{DELperturb}
\end{equation}
As a corollary of Proposition \ref{prop31}, we readily see that the operators in the l.h.s. of (\ref{DELperturb}) converge to the operators of the l.h.s. of (\ref{CELperturb}), as $\varepsilon$ tends to 0, provided the coefficients $(\gamma_\ell)$ satisfy the assumptions of the proposition. Moreover, if we consider any family $(A_i(\varepsilon,t),B_i(\varepsilon,t))_{1\leq i\leq n}$ of $2n$ functions from $]-\varepsilon_0,\varepsilon_0[\times[t_0,t_f]$ to $\mathbb{R}$, and if we set $a_i(t)=A_i(0,t)$ and $b_i(t)=B_i(0,t)$, then the formulas (\ref{convV}) and (\ref{convW}) imply that both sides of each equation in (\ref{DELperturb}) converge to the respective quantities in (\ref{CELperturb}).

\subsection{The Legendre transform in the classical and discrete settings}

Let us construct first the canonical coordinates in the rotating frame. Obviously, the decompositions in the galilean frame (\ref{LUTH1}) and (\ref{LUTH2}) do not longer hold when working in a rotating frame since there appears some inertial forces and effects.\\
Let us consider the classical case. As in classical textbooks (for instance \cite[p.~266]
{BP}), we choose as coordinates $a_i(t),b_i(t)$ and as momenta 
\begin{equation}
c_i(t)=m_i(\dot {a_i}(t)-\omega b_i(t)),~~~~d_i(t)=m_i(\dot{b_i}(t)+\omega a_i(t)).
\label{cd}
\end{equation} 
By using the derivatives of $x_{i1}(t)$ and $x_{i2}(t)$ obtained from (\ref{homographicsol}) and the expression of $T$ given in (\ref{TU}), we easily find 
\begin{equation*}
\mathcal{L}_c=\displaystyle \frac 12\sum_{i=1}^nm_i[(\dot {a_i}-\omega b_i)^2+(\dot {b_i}+\omega a_i)^2]+U(\mathbf{x})=
\displaystyle \frac{1}{2}\sum_{i=1}^n\frac{1}{m_i}(c_i^2+d_i^2)+U(\mathbf{x}).
\end{equation*}
The Lagrangian $\mathcal{L}_c$ depends naturally on the variables $a_i,b_i,\dot{a}_i,\dot{b}_i$ while the hamiltonian function $\mathcal{H}_c$, obtained through the Legendre transform of $\mathcal{L}_c$, depends essentially on $a_i,b_i,c_i,d_i$. Its value is given by
\begin{center}
$\mathcal{H}_c=\displaystyle \sum_{i=1}^n(\dot{a_i}\frac{\partial 
\mathcal{L}_c}{\partial \dot{a_i}}+\dot{b_i}\frac{\partial \mathcal{L}_c}{\partial \dot{b_i}})-\mathcal{L}_c=\displaystyle \frac{1}{2}\sum_{i=1}^n\frac{1}{m_i}(c_i^2+d_i^2)-\omega\sum_{i=1}^n(a_id_i-b_ic_i)-U(\mathbf{x})$.
\end{center}
By using (\ref{CELperturb}) and (\ref{cd}), we may easily show that the partial derivatives of $\mathcal{H}_c$ w.r.t. $c_i,d_i,a_i,b_i$ are respectively equal to
\begin{equation}
\dot {a_i}=\frac{\partial \mathcal{H}_c}{\partial c_i},~\dot {b_i}=\frac{%
\partial \mathcal{H}_c}{\partial d_i},~~\dot {c_i}=-\frac{\partial \mathcal{H%
}_c}{\partial a_i},~\dot {d_i}=-\frac{\partial \mathcal{H}_c}{\partial b_i}.
\label{CELhamiltonien}
\end{equation}
We consider now the discrete case. Plugging the various equations (\ref{homographicsol}) in (\ref{op1}) and (\ref{op2}), we get
\begin{center}
$\Box x_{i1}(t)=(V_c(A_i)-V_s(B_i))(t)\cos (\omega t)-(V_s(A_i)+V_c(B_i))(t)\sin (\omega t)$,
$\Box x_{i2}(t)=(V_s(A_i)+V_c(B_i))(t)\cos (\omega t)+(V_c(A_i)-V_s(B_i))(t)\sin (\omega t)$.
\end{center}
Now, squaring, expanding and summing, we find that the discrete Lagrangian defined by (\ref{LUTH2}) is given as follows 
\begin{equation*}
\mathcal{L}_d=\displaystyle \frac
12\sum_{i=1}^nm_i[(V_c(A_i)-V_s(B_i))^2+(V_s(A_i)+V_c(B_i))^2]+U(\mathbf{x}).
\end{equation*}
Because of the formal similarity between $\mathcal{L}_c$ and $\mathcal{L}_d$, we choose as canonical coordinates $A_i(t),B_i(t)$, and as momenta 
\begin{equation}
C_i(t)=m_i(V_c(A_i)(t)-V_s(B_i)(t)),~~~~D_i(t)=m_i(V_s(A_i)(t)+V_c(B_i)(t)).
\label{CD}
\end{equation}
Hence, $\mathcal{L}_d$ may be rewritten as
\begin{center}
$\mathcal{L}_d=\displaystyle \frac{1}{2}\sum_{i=1}^n\frac{1}{m_i}(C_i^2+D_i^2)+U(\mathbf{x})$.
\end{center}
The Lagrangian $\mathcal{L}_d$ depends naturally on the variables $A_i,B_i$, but also on $V_c(A_i)$, $V_c(B_i)$, $V_s(A_i)$, $V_s(B_i)$. In contrast, $\mathcal{L}_d$ does not depend naturally on the variables $C_i$ and $D_i$ introduced a posteriori nor on the variables $\Box A_i$ and $\Box B_i$. This is a clear indication of the non-covariance of the discretization procedure when dealing with inertial frames. At this point, it is essential to note that the Legendre transform may not be generalized in a convenient way to the discrete case. A first reason for this, is that the derivation of $\mathcal{L}_d$ w.r.t. the variables $\Box A_i$ and $\Box B_i$ is a nonsense. A second deeper one is that the discrete Euler-Lagrange equations are not covariant w.r.t. change of variables. A last reason is that the Hamilton's principle is not covariant in the discrete setting.

\subsection{Discrete hamiltonian equations in the rotating frame}

To overcome the difficulty mentioned previously, we introduce by analogy to the classical case the discrete hamiltonian function
\begin{center}
$\displaystyle\mathcal{H}_d=\frac{1}{2}\sum_{i=1}^n\frac{1}{m_i}(C_i^2+D_i^2)-\omega\sum_{i=1}^n(A_iD_i-B_iC_i)-U(\mathbf{x})$.
\end{center}
This construction has five interesting features. The first one is obviously that $\mathcal{H}_d$ depends in an algebraic way of the variables $A_i$, $B_i$, $C_i$, $D_i$, and not of some additional derivative operators acting on the previous variables. Next, if $\omega=0$, we recover (\ref{LUTH2}).
Another important feature is the possibility to provide the Hamilton-Jacobi partial differential equation in the discrete calculus of variation expressed as
\begin{center}
$\displaystyle\frac{1}{2}\sum_{i=1}^n\frac{1}{m_i}\left(\left(\frac{\partial S}{\partial A_i}\right)^2+\left(\frac{\partial S}{\partial B_i}\right)^2\right)-\omega\sum_{i=1}^n\left(A_i\frac{\partial S}{\partial B_i}-B_i\frac{\partial S}{\partial B_i}\right)-U(\mathbf{x})=cst$,
\end{center}
for the unknown action $S=S((A_i,B_i)_i)$. As one knows, this equation is of a crucial importance when constructing variational integrators for approximating the solutions of the equations of motion. The last two properties are given in the two following results.

\begin{proposition}
The equations of motion in lagrangian form (\ref {DELperturb}) are equivalent to 
\begin{eqnarray}
V_c(A_i)-V_s(B_i)+\omega B_i=\frac{\partial \mathcal{H}_d}{\partial C_i}, \label{DELhamiltonien1}\\
V_s(A_i)+V_c(B_i)-\omega A_i=\frac{\partial \mathcal{H}_d}{\partial D_i}, \label{DELhamiltonien2}\\
m_i(W_c(A_i)-W_s(B_i))-\omega D_i=\frac{\partial \mathcal{H}_d}{\partial A_i}%
, \label{DELhamiltonien3}\\
m_i(W_s(A_i)+W_c(B_i))+\omega C_i=\frac{\partial \mathcal{H}_d}{\partial B_i}%
.  \label{DELhamiltonien4}
\end{eqnarray}
Moreover, let us suppose that the operator $\Box$ satisfies (\ref{ourbox}) and (\ref{convcond}). Then for all $\varepsilon_0>0$ and for all set of $4n$ functions $(A_i,B_i,C_i,D_i)_{1,\leq i\leq n}:]-\varepsilon_0,\varepsilon_0[\times[t_0,t_f]\rightarrow\mathbb{R}$, continuous w.r.t. $(\varepsilon,t)$ and $\mathcal{C}^2$ w.r.t. $t$, the Lagrangian $\mathcal{L}_d$, the Hamiltonian $\mathcal{H}_d$ and the four
Hamilton's equations (\ref{DELhamiltonien1}) to (\ref{DELhamiltonien4}) converge locally uniformly in $]t_0,t_f[$ respectively to the Lagrangian $\mathcal{L}_c$, the Hamiltonian $\mathcal{H}_c$ and the four Hamilton's equations (\ref{CELhamiltonien}) as $\varepsilon$ tends to 0.
\end{proposition}

\begin{proof}
The two first equations (\ref{DELhamiltonien1}) and (\ref{DELhamiltonien2}) arise from the computation of the partial derivatives of $\mathcal{H}_d$ w.r.t. $C_i$ or $D_i$ and the definition (\ref{CD}) of these momenta. The two last ones  (\ref{DELhamiltonien3}) and (\ref{DELhamiltonien4}) are easy consequences of the values of the partial  derivatives of $\mathcal{H}_d$ w.r.t. $A_i$ or $B_i$ and the use of the r.h.s. of (\ref{DELperturb}) to eliminate  $\frac{\partial U}{\partial A_i}$ and $\frac{\partial U}{\partial B_i}$.\\
Now let us prove that the respective left hand-sides of equations (\ref{DELhamiltonien1}) to (\ref{DELhamiltonien4}) converge to the left hand-sides of equations (\ref{CELhamiltonien}) as $\varepsilon$ tends to 0, which depends mainly on (\ref{convV}) and (\ref{convW}). Let $(A,B,C,D)$ be four functions of class $\mathcal{C}^1$ w.r.t. $t$ and let us denote $(a(t),b(t),c(t),d(t))=(A(0,t),B(0,t),C(0,t),D(0,t))$. We get
\begin{center}
\small $V_c(A)-V_s(B)+\omega B\rightarrow \dot{A}(0,t)-\omega B+\omega B=\dot{a}$,\\
\small $V_s(A)+V_c(B)-\omega A\rightarrow \omega A+\dot{B}(0,t)-\omega A=\dot{b}$,\\
\small $m(W_c(A)-W_s(B))-\omega D\rightarrow m(\omega^2 A(0,t)-\ddot{A}(0,t)+2\omega \dot{B}(0,t))-\omega D(0,t)=-\dot{c}$,\\
\small $m(W_s(A)+W_c(B))+\omega C\rightarrow m(-2\omega \dot{A}(0,t)+\omega^2 B(0,t)+\ddot{A}(0,t))+\omega C(0,t)=-\dot{d}$
\end{center}
as $\varepsilon$ tends to $0$. From this we deduce that the schemes (\ref{DELhamiltonien1}) to (\ref{DELhamiltonien4}) converge to (\ref{CELhamiltonien}) and the analogous property for Lagrangian and Hamiltonian is obvious.
\end{proof}

\begin{remark}\rm
Although the discrete hamiltonian equations in the cartesian frame
look very similar to the classical ones (see formula (\ref{eqHamdisc})),
they do not behave covariantly under general change of coordinates since we
might expect, in a rotating frame, equations of the shape
\begin{center}
\textrm{$\Box A_i=\frac{\partial \mathcal{H}_d}{\partial C_i},~\Box B_i=%
\frac{\partial \mathcal{H}_d}{\partial D_i},~~-\Box ^{\star }C_i=-\frac{%
\partial \mathcal{H}_d}{\partial A_i},~-\Box ^{\star }D_i=-\frac{\partial 
\mathcal{H}_d}{\partial B_i}$. }
\end{center}
which are not true. Hence, hamiltonian discrete equations are not covariant in general.
\label{rem41}
\end{remark}

\begin{remark}\rm
Let us suppose that the functions $A_i,B_i$ express some lengths, then both sides of (\ref{DELhamiltonien1}) and (\ref{DELhamiltonien2}) are celerities, i.e. meters/seconds, and both sides of (\ref{DELhamiltonien3}) and (\ref{DELhamiltonien4}) are forces, i.e. Newton.
\end{remark}

\section{Relative equilibria solutions to the generalized $n$-body problem in classical and discrete settings}

We recall that a relative equilibrium solution of a generalized $n$-body
problem is a configuration of $n$ moving particles which are located at
fixed points in a uniformly rotating plane. We mention the terminology used in \cite[pp. 217, pp. 219]{BP} as planar
solution and homographic solution with dilatation equal to 1. Some other
authors call those solutions Lagrange configurations.

We shall study first the case where $(a_i,b_i)$ and $(A_i,B_i)$ are constant
w.r.t. time. We shall give some remarks at the end of this section in the
case when the coordinates $(a_i,b_i)$ and $(A_i,B_i)$ are of the shape $%
(a_i,b_i)=(a_i^0\lambda (t),b_i^0\lambda (t))$ and $(A_i,B_i)=(A_i^0\Lambda
(t),B_i^0\Lambda (t))$ where $\lambda $ and $\Lambda $ are some dilatation
factors. Note that the mutual distances $r_{ij}$ and $R_{ij}$ are constant w.r.t.
time. 

\subsection{Existence of equilibria in galilean frames}

As usually done in the classical case, one may ask if there exist solutions of (\ref{mcelter}) which are constant
w.r.t. time. 

\begin{proposition}
Let us consider a galilean frame, let $I$ be an interval included in $\mathbb{R}$. Let us suppose that one of the two following conditions holds :
\begin{itemize}
\item $[t_0,t_f]\subset I$ and $\Box $ is not defined by $(0,\ldots,0,\gamma_0,0,\ldots,0)\in\mathbb{R}^{2N+1}$,
\item $I=\mathbb{R}$ and $\Box\neq 0$.
\end{itemize}
Then there does not exist solutions of (\ref{mcelter}) remaining constant w.r.t time in the interval $I$.
\label{prop51}
\end{proposition}

\begin{proof}
Let us prove the two points by contraposition. Let us suppose that there exists a solution $\{x_{ik}(t)\}$, $i$ and $k$ running from 1
to $n$ and $d$ respectively, of (\ref{mcelter}) remains constant w.r.t. time
inside $I$. Then formula (\ref{mcelter}) shows that $\Box ^{\star }\Box 1$
must be constant, say $cst$, over $I$. In order to compute the value
of $\Box ^{\star }\Box f(t)$ for any function $f(t)$, we use (\ref{mchomat2}%
) and apply it to $f(t)=1$. We look for the coefficients $(\gamma _\ell
)_\ell $ in order to check $\Box ^{\star }\Box 1=cst$ in $I$. In both cases, we shall assume for ease of exposition that $N=1$ since the proof for arbitrary $N$ is similar . 

In the first case, when $[t_0,t_f]\subset I$, we consider the five explicit values of $\Box ^{\star }\Box 1$ in the
convenient intervals
\begin{itemize}
\item  if $t\in [t_0,t_0+\varepsilon[$, $\Box ^{\star }\Box
1(t)=\frac 1{\varepsilon ^2}(\gamma _{-1}(\gamma _{-1}+\gamma
_0+\gamma _1)+\gamma _0(\gamma _0+\gamma _1))$,
\item  if $t\in [t_0+\varepsilon ,t_0+2\varepsilon [$, $\Box ^{\star }\Box
1(t)=\frac 1{\varepsilon ^2}((\gamma _{-1}+\gamma _0)(\gamma _{-1}+\gamma
_0+\gamma _1)+\gamma _1(\gamma _0+\gamma _1))$,
\item  if $t\in [t_0+2\varepsilon ,t_f-2\varepsilon ]$, $\Box ^{\star }\Box
1(t)=\frac 1{\varepsilon ^2}(\gamma _{-1}+\gamma _0+\gamma _1)^2=(\Box
1(t))^2$,
\item  if $t\in ]t_f-2\varepsilon ,t_f-\varepsilon ]$, $\Box ^{\star }\Box
1(t)=\frac 1{\varepsilon ^2}(\gamma _{-1}(\gamma _{-1}+\gamma _0)+(\gamma
_0+\gamma _1)(\gamma _{-1}+\gamma _0+\gamma _1))$,
\item  if $t\in ]t_f-\varepsilon,t_f]$, $\Box ^{\star }\Box
1(t)=\frac 1{\varepsilon ^2}(\gamma _{0}(\gamma _{-1}+\gamma _0)+\gamma _1(\gamma _{-1}+\gamma _0+\gamma _1))$.
\end{itemize}
The second equation being identical to the fourth one, one sees easily that the system implies $(\gamma_{-1},\gamma_0,\gamma_1)=(0,\gamma_0,0)$. The case for arbitrary $N$ is similar.

Now, let us deal with the second case $I=\mathbb{R}$. We recall that $%
\Box \mathbf{x}$ is compactly supported for all $\mathbf{x}$ (see 
\cite[Proposition 2.1]{RS1}). So we see that the l.h.s. of (\ref{mcelter})
vanishes, for all index $i$, outside the interval $[t_0-N\varepsilon
,t_f+N\varepsilon ]$. Since the functions appearing in the rh.s. of (\ref
{mcelter}) are obviously constant, the l.h.s. of (\ref{mcelter}) vanishes
for all $t\in \mathbb{R}$. As a rule, since $\Box ^{\star }\Box
1(t)=(\Box 1(t))^2$ for $t\in [t_0+2N\varepsilon ,t_f-2N\varepsilon]$, we must have $\sum_\ell \gamma_\ell=0$ or equivalently $\gamma_0=-\sum_{\ell\neq 0}\gamma_\ell$ which implies that $\Box=0$ from the previous result, thus contradicting the assumption and this ends the proof.
\end{proof}

\begin{remark}\rm
The study of constant solutions in the newtonian case is much more simple since the functions $f_{ij}$ are all increasing or all decreasing. Indeed, we may prove the result by considering for all $k\in\{1,\ldots,d\}$ the equations (\ref{celuniv}) or (\ref{mcelter}) of index $i\in\{1,\ldots,n\}$ maximizing $x_{ik}$, without 
studying $\Box^\star\Box$ in the various intervals of time.
\label{rem51}
\end{remark}

\begin{remark}\rm
In contrast with Proposition \ref{prop51} and Remark \ref{rem51}, it might exist constant solutions of (\ref{celuniv}) in $\mathbb{R}$. This occurs for instance when considering the Laplace-Sellinger or the London potentials $f_{ij}$.
\end{remark}

\subsection{The algebraic equations of relative equilibria}

We introduce the system of $2n$ algebraic equations 
\begin{equation}
-\lambda m_ix_i=\sum_{j\neq i}{f}_{ij}'(s_{ij})\frac{x_i-x_j}{s_{ij}}
\mbox{ and }-\lambda m_iy_i=\sum_{j\neq i}{f}_{ij}'(s_{ij})\frac{y_i-y_j}{s_{ij}},  
\label{algeq}
\end{equation}
where $s_{ij}=((x_i-x_j)^2+(y_i-y_j)^2)^{1/2}$, which constitute a slight generalization of the algebraic equations for relative equilibria to appear later on. The unknowns are the $2n+1$ real numbers $x_i$, $y_i$ and $\lambda$. The number $\lambda$ is related to the pulsation $\omega$ of the configuration. The $2n$ preceding equations are dependent because they sum to 0 so that $\sum_{k=1}^nm_kx_k=\sum_{k=1}^nm_ky_k=0$ using an argument of symmetry. We note that if all the functions $f_{ij}(r)$ are algebraic w.r.t. $r$ then the functions ${f}_{ij}'(r)/r$ are also algebraic and thus, the equations (\ref{algeq}) are algebraic w.r.t. the coordinates $a_i,b_i$
and $A_i,B_i$. In constrast, those equations are not algebraic w.r.t. $\varepsilon$ or $\omega$. We note also by the way that equations (\ref{algeq}) are not invariant by translation. However, they are invariant by rotation in the plane, that is to say if $(x_i,y_i)$ is a set of solutions, then for all $\alpha\in\mathbb{R}$, 
$(x_i\cos\alpha-y_i\sin\alpha,x_i\sin\alpha+y_i\cos\alpha)$ is another set of solutions of (\ref{algeq}). The problem of finiteness of the quotient set of solutions by the orthogonal group $SO(2,\mathbb{R})$ remains open even in the newtonian case and is known as the Wintner's conjecture mentioned in \cite{SS}. We may conjecture, help to Bezout theorem in algebraic geometry, that (\ref{algeq}) is a system of algebraic equations of rank $2n-2$ with no common zero-hypersurfaces and
thus have finitely many solutions up to rotations.
\begin{proposition}
Let $n=3$. We assume that the three potential functions $f_{ij}$ ($i,j\in\{1,2,3\}$, $i<j$) satisfy the following condition : there exists an injective function $\zeta:\mathbb{R}^\star_+\rightarrow \mathbb{R}$ such that for all $s>0$, one has 
\begin{center}
${f}_{ij}'(s)=s\frac{\zeta(s)}{m_k}$, where $\{i,j,k\}=\{1,2,3\}$.
\end{center} 
Then each configuration $(x_i,y_i)$ satisfying (\ref{algeq}) for some $\lambda$ is either colinear or equilateral.
\end{proposition}
\begin{proof}
Suppose that $(x_i,y_i)$ is a solution of (\ref{algeq}) for some $\lambda$. We deduce from (\ref{algeq}) the following equations
\begin{equation}
\begin{pmatrix}
m_1 & m_2\\
\displaystyle \frac{{f}_{13}'(s_{13})}{s_{13}} & \displaystyle \frac{{f}_{23}'(s_{23})}{s_{23}}
\end{pmatrix}\begin{pmatrix}x_1-x_3\\x_2-x_3\end{pmatrix}=
\begin{pmatrix}-(m_1+m_2+m_3)x_3\\\lambda m_3x_3
\end{pmatrix}
\label{mat}
\end{equation}
and an entirely similar system for the vector $\tra{(y_1-y_3,y_2-y_3)}$. By permuting the indices $1,2,3$, we obtain six bidimensional linear systems, inducing only three different matrices. Let us suppose that one of these three matrices is regular, say the one occuring in (\ref{mat}). Then we obtain $(x_1,x_2)$ and $(y_1,y_2)$ as functions of $x_3$ and $y_3$ respectively and we check easily that the four points $(x_i,y_i)$ and $(0,0)$ lie on the same straight line. Now, when all the three matrices are singular, one has the following system
\begin{center}
$\left\{\begin{array}{l}
{f}_{13}'(s_{13})m_2s_{23}={f}_{23}'(s_{23})m_1s_{13}\\
{f}_{12}'(s_{12})m_3s_{23}={f}_{23}'(s_{23})m_1s_{12}\\
{f}_{13}'(s_{13})m_2s_{12}={f}_{12}'(s_{12})m_3s_{13}
\end{array}\right.$
\end{center}
which yields simply $\zeta(s_{12})=\zeta(s_{13})=\zeta(s_{23})$. Due to the assumptions on $\zeta$, we get $s_{12}=s_{13}=s_{23}$.
\end{proof}
\begin{remark}\rm
In contrast with the proof given in \cite{BP}, we do not use Galileo's law.
\end{remark}
\begin{remark}\rm
If we choose for some fixed number $\beta\in\mathbb{R}^\star$, $\zeta(s)=m_1m_2m_3\beta s^{\beta-2}$, we recover the homogeneous potential function occuring in Section 1.
\end{remark}

Now, let us connect the system (\ref{algeq}) to the search of constant solutions of (\ref{CELperturb}), respectively (\ref{DELperturb}). We shall say that a solution $\{(a_i,b_i)\}$ of (\ref{CELperturb}) is a relative equilibrium if all coordinates $(a_i,b_i)$ are independent on $t\in\mathbb{R}$. Similarly, we say that a solution $\{(A_i,B_i)\}$ of (\ref{DELperturb}) is a relative equilibrium if all coordinates $(A_i,B_i)$ are independent on $t\in[t_0+2N\varepsilon,t_f-2N\varepsilon]$. This specific interval is chosen in such a way that relative equilibria exist in each setting and are closely connected. Let us note that, in the case the previous interval is replaced with $\mathbb{R}$, no solution would be found as an analysis similar to the proof of Proposition \ref{prop51} shows. 

We note that the function $W_c(1)(t)$ takes a constant value inside $[t_0+2N\varepsilon,t_f-2N\varepsilon]$. We assume in the remainder of this paper that this constant is positive and we denote it by $\Omega^2(\varepsilon)$. Let us remark that if $\Box$ satisfies (\ref{convcond}), then $\lim_{\varepsilon\to 0}\Omega^2(\varepsilon)=\omega^2$ locally uniformly in $]t_0,t_f[$, as the  formula (\ref{convW}) in Proposition \ref{prop31} shows, and the previous assumption is satisfied. 

We observe that, in order that a configuration of $n$ particles $\{(a_i,b_i)\}$ is a relative
equilibrium solution of (\ref{CELperturb}), it is necessary and sufficient that the set $%
\{(a_i,b_i)\}$ is solution of (\ref{algeq}) with $\lambda =\omega ^2$. Indeed, this is a simple consequence of plugging constant functions $(a_i,b_i)$ in (\ref{CELperturb}). In a similar way, we obtain the following
\begin{proposition}
In order that a configuration of $n$ particles $\{(A_i,B_i)\}$ is a relative
equilibrium solution of (\ref{DELperturb}), it is necessary and sufficient that the set $%
\{(A_i,B_i)\}$ is solution of (\ref{algeq}) with $\lambda=\Omega^2(\varepsilon)$ 
\label{prop4.1}
\end{proposition}
\begin{proof} 
We have in $[t_0+2N\varepsilon,t_f-2N\varepsilon]$
\begin{center}
$\displaystyle W_c(1)=\Omega^2(\varepsilon)=\frac{1}{\varepsilon ^2}\hspace{-0.2cm}\sum_{\tiny 
\begin{array}[t]{c}
|\ell|\leq 2N\\
|j|\leq N\\ 
|\ell+j|\leq N
\end{array}
}\hspace{-0.5cm}\gamma _{\ell+j}\gamma _j\cos
(\ell\varepsilon \omega )\mbox{, }\displaystyle W_s(1)=\frac{1}{\varepsilon^2}\hspace{-0.2cm}\sum_{\tiny 
\begin{array}[t]{c}
|\ell|\leq 2N\\
|j|\leq N\\ 
|\ell+j|\leq N
\end{array}
}\hspace{-0.5cm}\gamma _{\ell+j}\gamma _j\sin
(\ell\varepsilon \omega )$.
\end{center}
By using an easy symmetry argument, we prove that $W_s(1)(t)=0$, for all $t\in[t_0+2N\varepsilon,t_f-2N\varepsilon]$. So, when we suppose that the functions $\{(A_i,B_i)\}$ remain constant in $[t_0+2N\varepsilon,t_f-2N\varepsilon]$, formula (\ref{DELperturb}) gives rise to (\ref{algeq}) with $\lambda=\Omega^2(\varepsilon)$.
\end{proof}

\subsection{Expansion factor and constants of motion for generalized $n$-body problem with homogeneous potential functions}
\label{sec53}

When the potentials are homogenous with exponent $\beta$, the solution sets of the two systems of algebraic equations (\ref{algeq}), obtained when  
$\lambda=\omega^2$ and $\lambda=W_c(1)$, are the same up an homothety. The following result shows this claim. 
\begin{proposition}
Let us suppose that $f_{ij}(r)=\mu _{ij}r^\beta $ with $\beta \neq 2$, then
to each configuration of $n$ bodies in relative equilibrium for (\ref{CELperturb})
corresponds an homothetic configuration in relative equilibrium for (\ref{DELperturb}) whose homothety ratio is the real number $\displaystyle \varphi (\varepsilon)=\left( \frac{\omega ^2}{\Omega^2(\varepsilon)}\right) ^{\frac 1{2-\beta }}$. Furthermore, the kinetic energies $T_C$ and $T_D$, the potential energies $U_C$ and $U_D$ and the angular momenta $\sigma_C$ and $\sigma_D$ of those two homothetic configurations are linked together as
\begin{center}
$T_D=\varphi(\varepsilon)^2T_C$,~~$U_D=\varphi(\varepsilon)^\beta U_C$ and $\sigma_D=\varphi(\varepsilon)^2\sigma_C$.
\end{center}
\end{proposition}
\begin{proof}
Indeed, we see that the two systems of equations (\ref{algeq}) obtained when $\lambda=\omega^2$ and $\lambda=\Omega^2(\varepsilon)$ may be rewritten respectively as 
\begin{equation*}
-\frac{\omega^2}{\beta}m_ia_i=\displaystyle\sum_{j\neq i}\mu _{ij}r_{ij}^{\beta -2}(a_i-a_j)%
\mbox{, }-\frac{\omega^2}{\beta}m_ib_i=\displaystyle\sum_{j\neq i}\mu _{ij}r_{ij}^{\beta
-2}(b_i-b_j)  
\end{equation*}
and
\begin{equation*}
-\frac{\Omega^2(\varepsilon)}{\beta}m_iA_i=\displaystyle\sum_{j\neq i}\mu _{ij}R_{ij}^{\beta -2}(A_i-A_j)%
\mbox{, }-\frac{\Omega^2(\varepsilon)}{\beta}m_iB_i=\displaystyle\sum_{j\neq i}\mu _{ij}R_{ij}^{\beta
-2}(B_i-B_j).  
\end{equation*}
Searching for solutions of the second system of the shape $A_i=a_i\varphi$ and $B_i=b_i\varphi$, both systems agree if and only if $\Omega^2(\varepsilon)\varphi(\varepsilon )^{2-\beta }=\omega ^2$ whence the value of $\varphi$.

Now, let us deal with the integrals of motion. Since $U$ is homogeneous of degree $\beta$, we have $U_D=\varphi(\varepsilon)^\beta U_C$. Next, we use formulas (\ref{TU}) and (\ref{homographicsol}) to compute $T_C=\omega^2 I_0$ and $T_D=\omega^2\varphi^2(\varepsilon)I_0$ where $I_0=\frac12\sum_im_i(a_i^2+b_i^2)$ is the moment of inertia. Lastly, the only nonzero component of the angular momentum tensor is equal to $\sigma_C=\sum_{i<j}\mu _{ij}(x_{i1}\dot x_{j2}-\dot x_{i1}x_{j2})=\omega\sum_{i<j}\mu _{ij}(a_ia_j+b_ib_j)$ and obviously $\sigma_D=\varphi^2(\varepsilon)\sigma_C$.
\end{proof}
The homothety ratio $\varphi(\varepsilon)$ will be called the expansion factor. For arbitrary $N$, the condition (\ref{convcond}) ensures that $\lim_{\varepsilon\to 0}\varphi(\varepsilon)=1$ but the converse does not hold. For example, when $N=1$, the expansion factor is equal to {\small 
\begin{equation*}
\varphi (\varepsilon )=\left( \frac{\omega ^2\varepsilon ^2}{(\gamma
_{-1}+\gamma _0+\gamma _1)^2+2\gamma _{-1}\gamma _1(\cos 2\omega \varepsilon
-1)+2\gamma _0(\gamma _{-1}+\gamma _1)(\cos \omega \varepsilon -1)}\right)
^{\frac 1{2-\beta }}.
\end{equation*}}
The existence of a finite nonzero limit to $\varphi(\varepsilon)$ as $\varepsilon$ tends to 0 is equivalent to the following two equations
\begin{center}
$\gamma _{-1}+\gamma_0+\gamma _1=0$ and $\gamma _0\gamma _{-1}+\gamma _0\gamma _1+4\gamma_{-1}\gamma _1=-1$.
\end{center} 
This system admits two families of solutions. The first one is a family of operators $\Box$ satisfying $\Box t=-1$ which do not check the condition (\ref{convcond}). The second one is an affine straight line of operators $\Box$ satisfying (\ref{convcond}) and given by 
\begin{equation}
\Box^{[r,s]}\mathbf{x}(t)=-\frac{s}{\varepsilon}\mathbf{x}(t-\varepsilon)\chi(t-\varepsilon)+\frac{s-r}{\varepsilon}\mathbf{x}(t)\chi(t)+\frac{r}{\varepsilon}\mathbf{x}(t+\varepsilon)\chi(t+\varepsilon)
\label{boxrs}
\end{equation}
together with the condition $r+s=1$, and that we have already encountered in \cite{RS1}. 

\section{Numerical experiments}

We present in the following the planar graphs associated to the restricted 3-body problem yielding a heavy, a light and a negligible bodies. The parameter $\mu$ stands for the normalized ratio between the lightest and the sum of the lightest and heaviest bodies. We choose to work only with the libration points $L_4$ and $L_5$ and not with the three unstable eulerian points $L_1,L_2,L_3$, see \cite{BP}. 
We consider an intermediate time $t_\nu\in[t_0+2N\varepsilon,t_f-2N\varepsilon]$ at which the particle $P_3$ is located at the neighbourhood of $L_4$ (or $L_5)$. We use some specific operators $\Box$ of the shape (\ref{boxrs}) and especially $\Box^{[1,0]}$, $\Box^{[0,1]}$, $\Box^{[\frac12,\frac12]}$ and $\Box^{[\frac{1-i}2,\frac{1+i}2]}$. 

Numerical experiments consist in solving (\ref{mcelter}) and (\ref{DELperturb}). Although these equations are functional ones, we solve them numerically by computing $A_3(t)$ and $B_3(t)$ on the grid $\{t_\nu+k\varepsilon,k\in\mathbb{Z}\}\cap[t_0+2N\varepsilon,t_f-2N\varepsilon]$.
 The convergence mode of operators $\Box$ in the function space of continuously differentiable functions on $[t_0,t_f]$ is locally uniform in $]t_0,t_f[$ and this induces numerical difficulties relative to the stability of the Cauchy problem at $t=t_0$ or $t=t_f$. This is the reason why the intermediate time $t_\nu$ has been introduced.

In order to compare the performances of each operator $\Box^{[r,s]}$ presented previously, we compute the error norm $err:=\|\mathbf{x}(t_\nu+M\varepsilon)-\mathbf{x}(t_\nu)\|_2$ with $M=5\times m$ and $m\in\{0,\ldots,100\}$ . We use for this equations (\ref{mcelter}) and the system (\ref{DELhamiltonien1})-(\ref{DELhamiltonien4}) which amounts to equations (\ref{DELperturb}), abbreviated respectively as DEL (Discrete Euler-Lagrange equations) and DHE (Discrete Hamiltonian Equations).\\

Most numerical experiments use the value $\mu=0.012$ associated to the system consisting of the Earth, the Moon and a rocket. The first one illustrates the fact that solving equations DHE give more accurate results than solving equations DEL, see Figure \ref{figure0}. In addition, we note that whenever the operators $\Box^{[1,0]}$ and $\Box^{[0,1]}$ are not convenient to solve (\ref{mcelter}), they become the best choice for solving (\ref{DELperturb}).
\begin{figure}[!ht]
\centering\includegraphics[width=9cm,height=6cm]{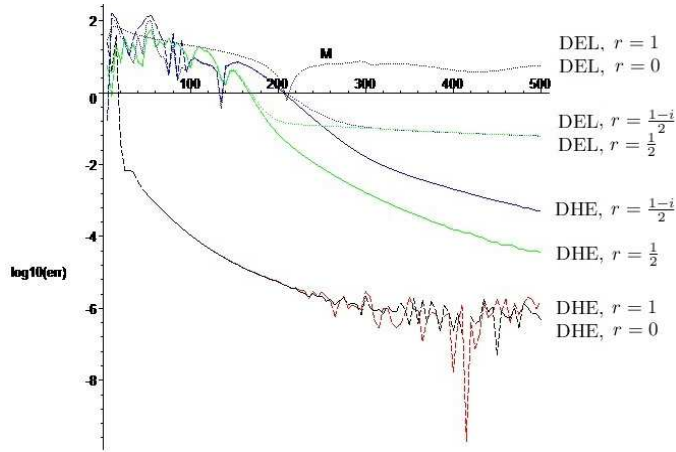}
\caption{2-norm of the error $\mathbf{x}(t_\nu+M\varepsilon)-\mathbf{x}(t_\nu)$}
\label{figure0}
\end{figure}

The second experiment uses the operator $\Box_q:=\Box^{[\frac{1-i}2,\frac{1+i}2]}$. The step number $m$ per period which induces the smallness of $\varepsilon$ is set to $m=15,30,50$ with $[t_\nu,t_\nu+M\varepsilon]=[0,5\pi]$. The six following figures \ref{figure1a} to \ref{figure1f} provide the trajectories of the three bodies when using equations DEL and DHE. The results with the three other operators $\Box$ are quite similar. 
\begin{figure}[!ht]
\begin{minipage}[c]{.46\linewidth}
\includegraphics[width=4cm]{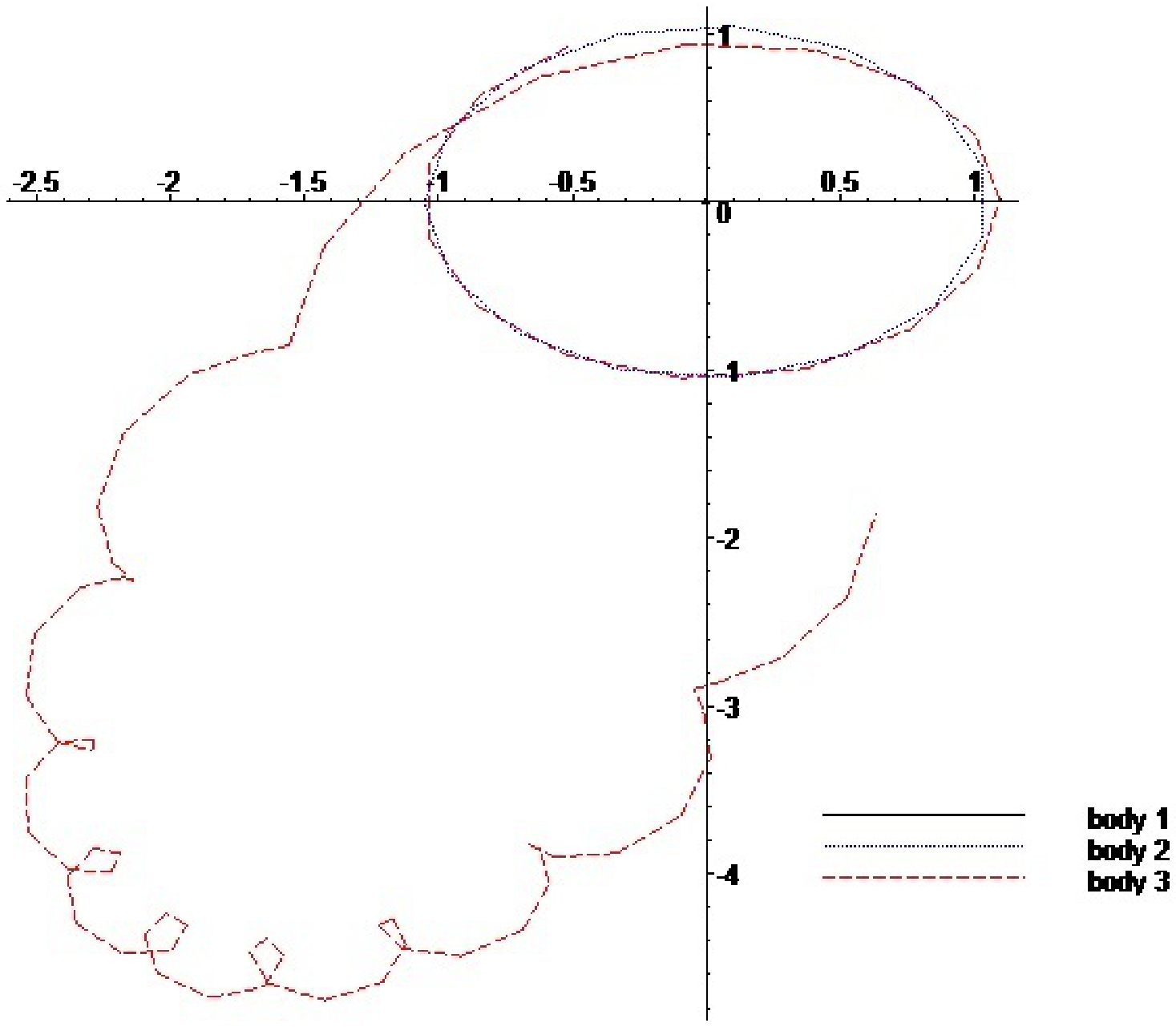}
\caption{DEL, $k=5$, $m=15$}
\label{figure1a}
\end{minipage} \hfill
\begin{minipage}[c]{.46\linewidth}
\includegraphics[width=6cm,height=4cm]{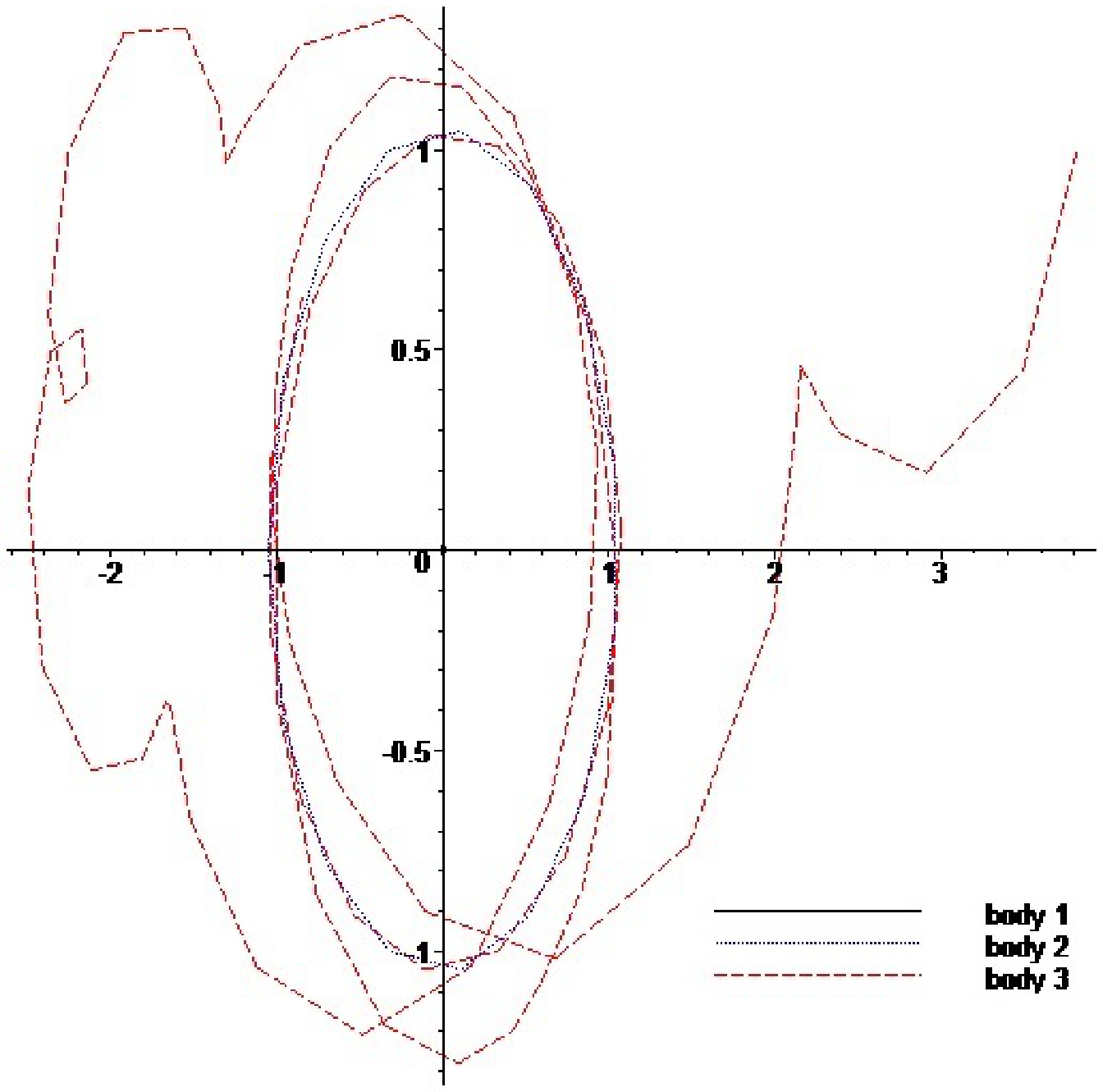}
\caption{DHE, $k=5$, $m=15$}
\label{figure1b}
\end{minipage}\hfill
\begin{minipage}[c]{.46\linewidth}
\includegraphics[width=4cm]{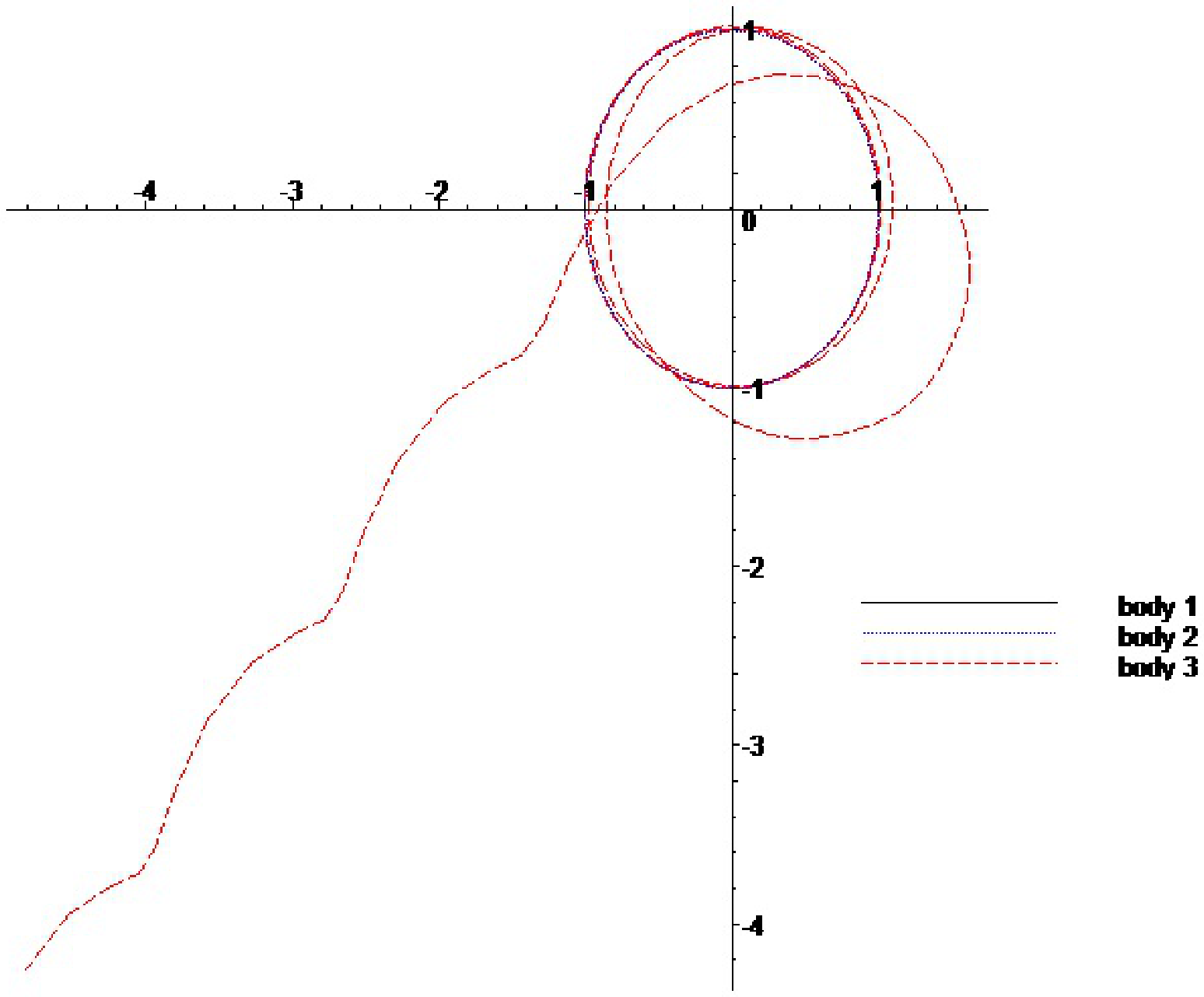}
\caption{DEL, $k=5$, $m=30$}
\label{figure1c}
\end{minipage}\hfill
\begin{minipage}[c]{.46\linewidth}
\includegraphics[width=4cm,height=3.3cm]{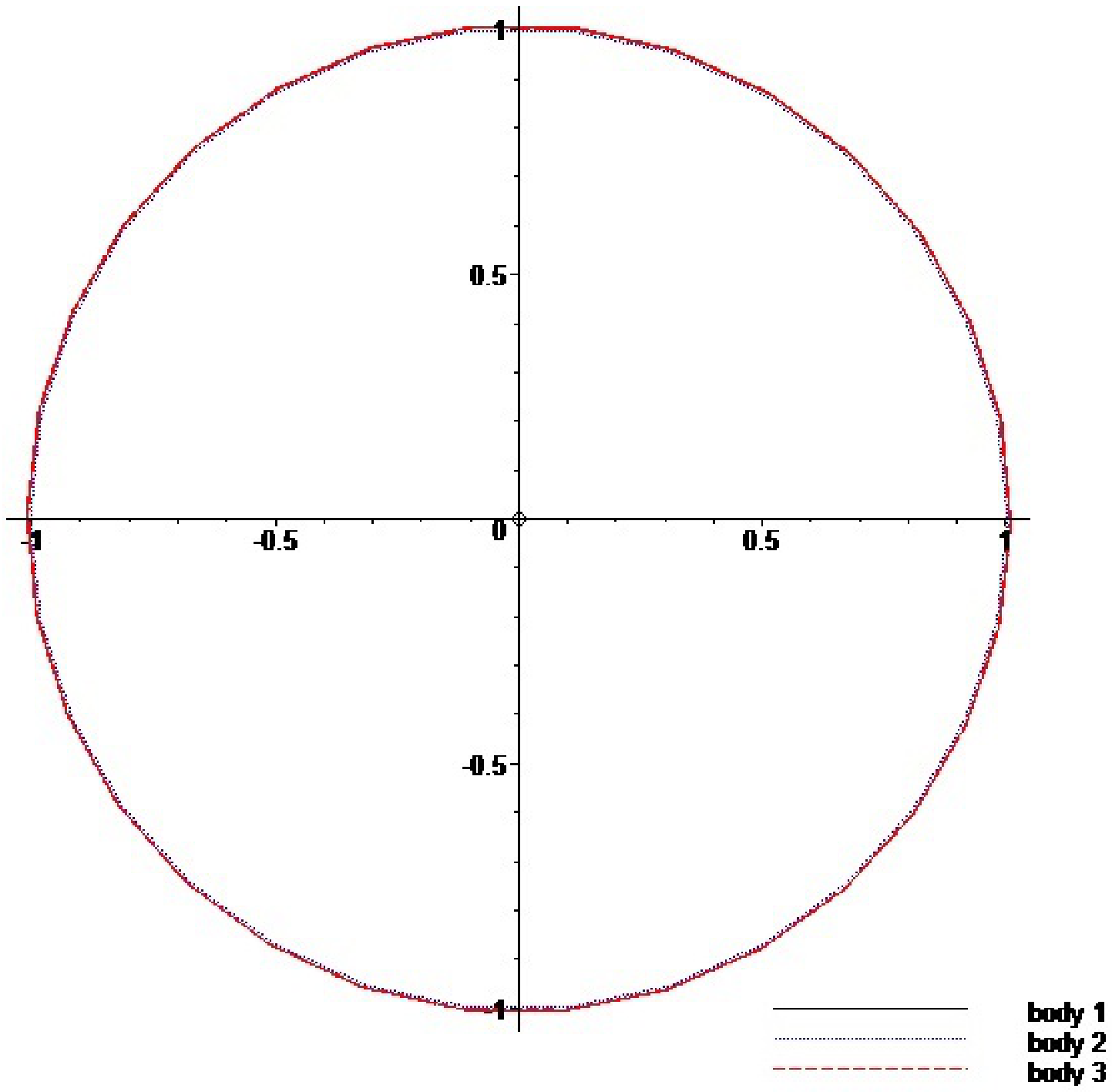}
\caption{DHE, $k=5$, $m=30$}
\label{figure1d}
\end{minipage}\hfill
\begin{minipage}[c]{.46\linewidth}
\includegraphics[width=4cm]{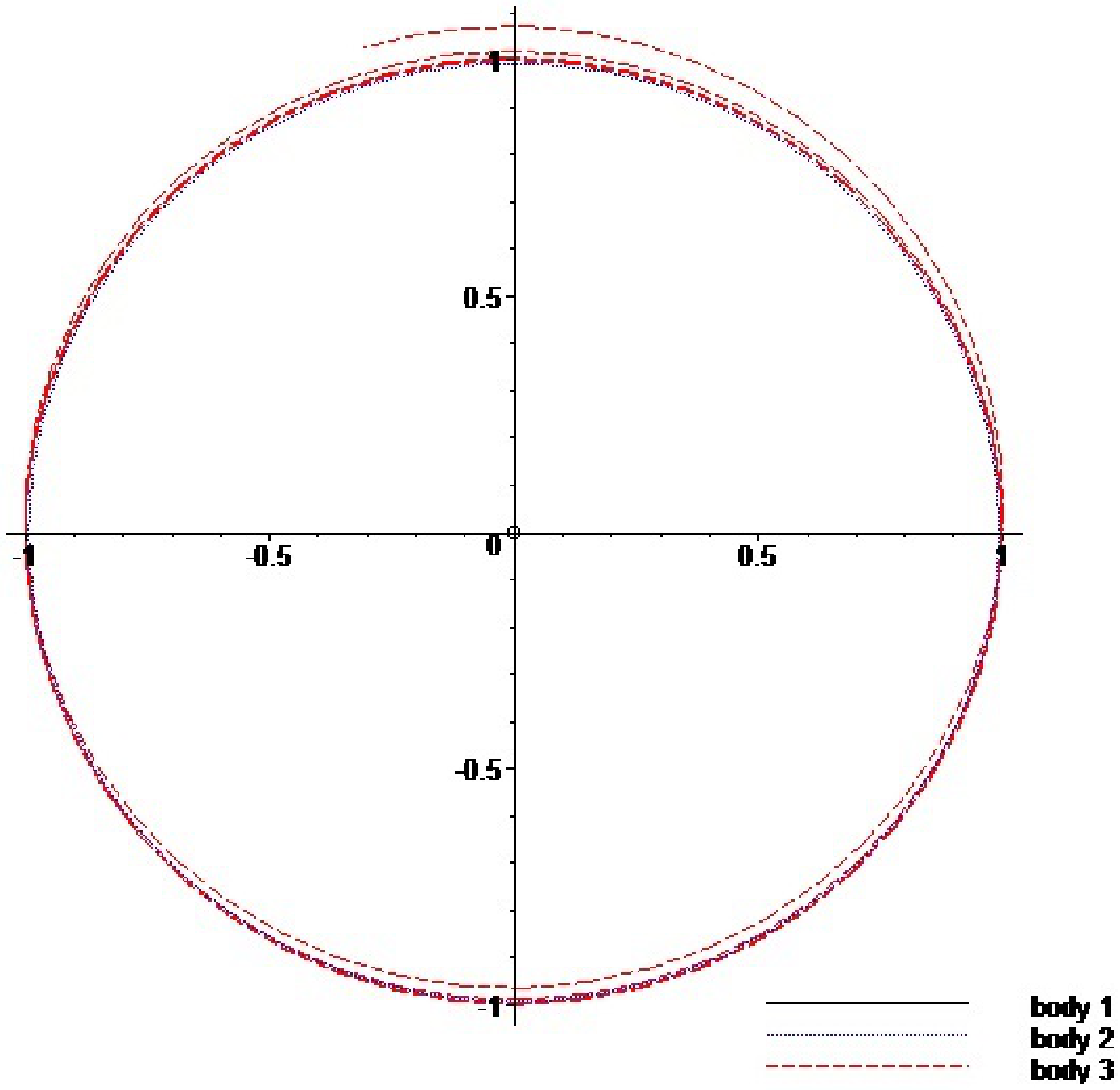}
\caption{DEL, $k=5$, $m=50$}
\label{figure1e}
\end{minipage}\hfill
\begin{minipage}[c]{.46\linewidth}
\includegraphics[width=4cm]{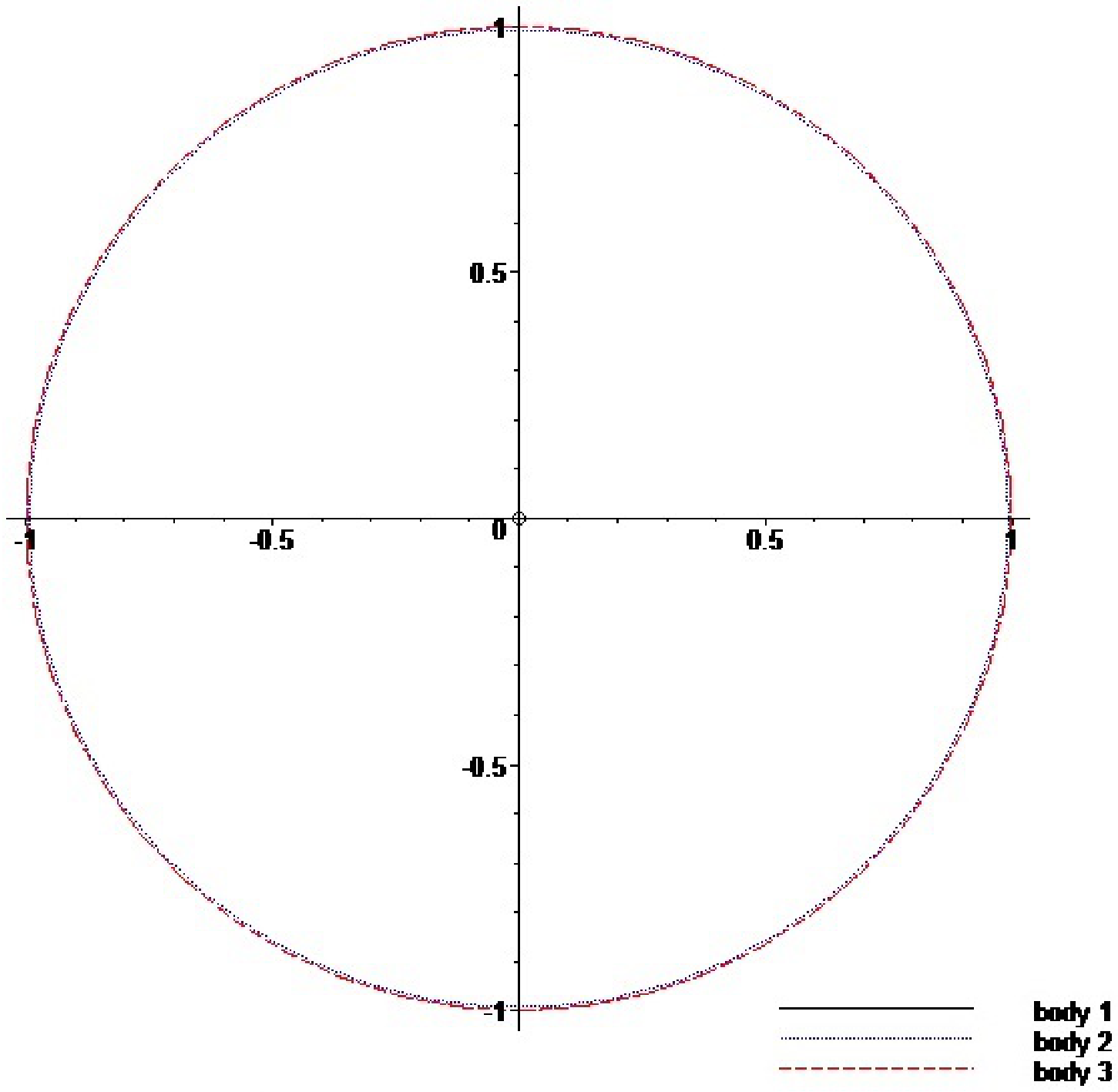}
\caption{DHE, $k=5$, $m=50$}
\label{figure1f}
\end{minipage}
\end{figure}
\noindent These figures illustrate the essential role played by $M$	and the obvious performance of DHE versus DEL. As we can see, 30 steps per period are necessary to reach an acceptable solution of the restricted three-body problem.

The next experiment highlights the crucial role played by the perturbations of the position at time $t_\nu$ of the lightest body, particularly in the equations (\ref{mcelter}). We still use the operator $\Box_q:=\Box^{[\frac{1-i}2,\frac{1+i}2]}$ and the system Earth-Moon-rocket. Using the framework of \cite{BP}, we perturb $A_3(0,t_\nu)=\mu-\frac12$ and $B_3(0,t_\nu)=\frac{\sqrt{3}}{2}$ respectively as $A_3(\varepsilon,t_\nu)=(\mu-\frac12)\varphi(\varepsilon)+\delta$ and $B_3(\varepsilon,t_\nu)=\frac{\sqrt{3}}2\varphi(\varepsilon)+\delta'$ and we choose $\delta=\delta'=0.01$ and next, $\delta=\delta'=0.05$. As we can see in Figures \ref{figure2a} and \ref{figure2b}, the trajectory of the rocket becomes more unstable as $\delta$ increases.
\begin{figure}[!ht]
\begin{minipage}[c]{.46\linewidth}
\includegraphics[width=4cm]{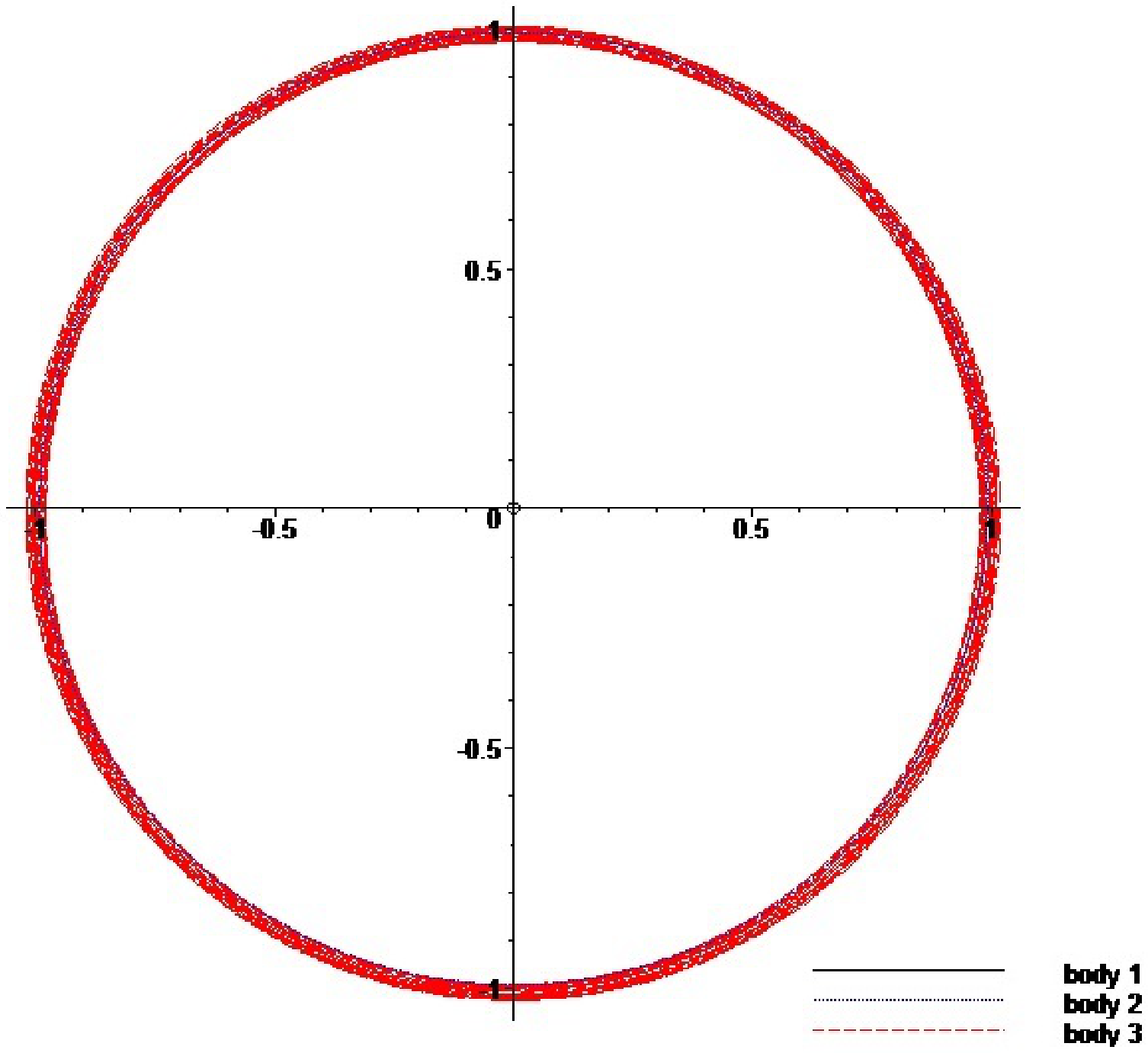}
\caption{DHE, $k=20$, $m=50$, $\delta=0.01$}
\label{figure2a}
\end{minipage} \hfill
\begin{minipage}[c]{.46\linewidth}
\includegraphics[width=4cm]{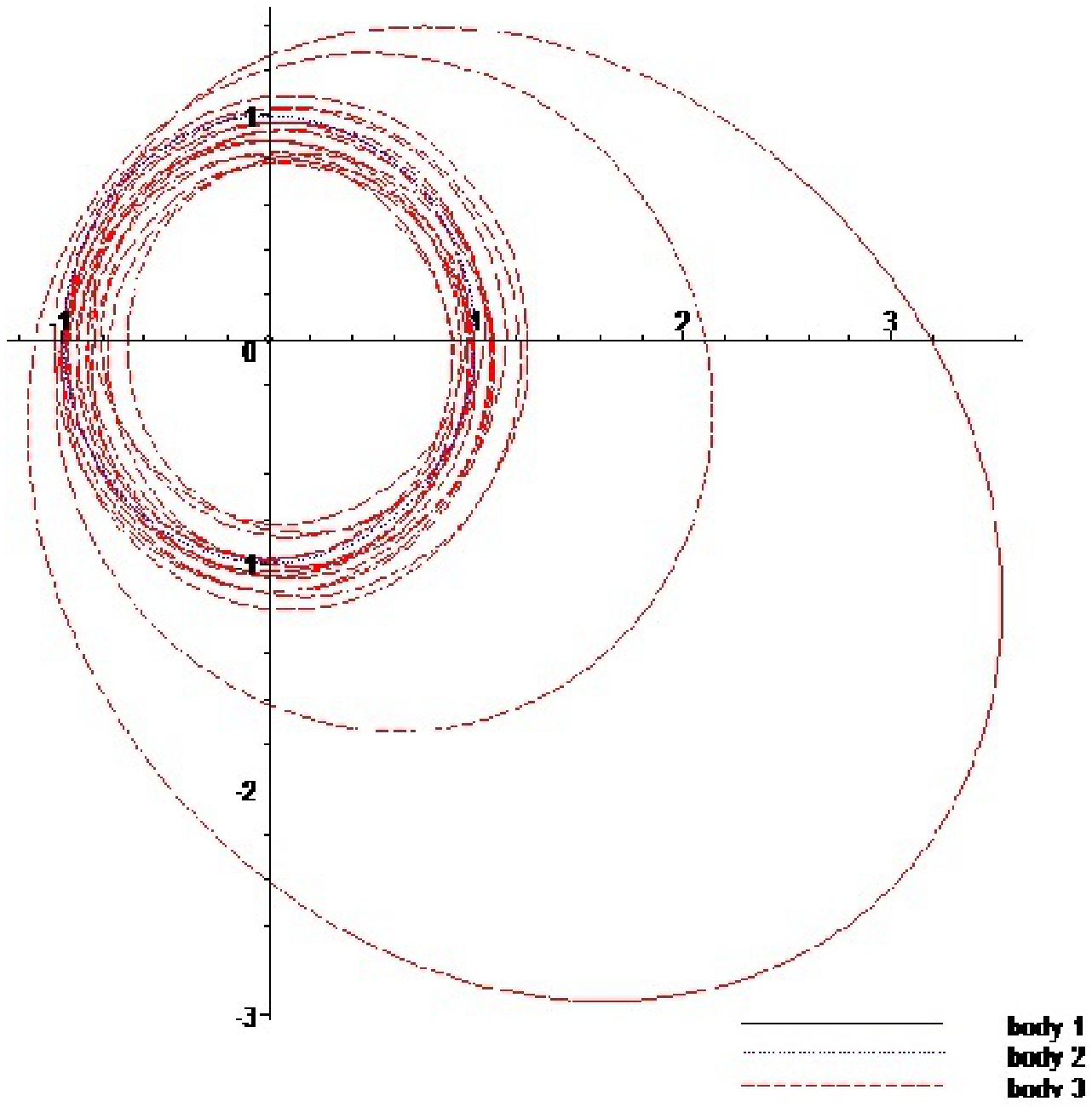}
\caption{DHE, $k=20$, $m=50$, $\delta=0.05$}
\label{figure2b}
\end{minipage}
\end{figure}

At last, we modify the ratio $\mu$ and give up the system earth-moon-rocket. The goal of this experiment is to illustrate the unstability of the system when $\mu>0.0385$ (see \cite{BP} for example). We use for this two values of $\mu$ which are greater than $0.04$.
\begin{figure}[!ht]
\begin{minipage}[c]{.46\linewidth}
\includegraphics[width=4cm]{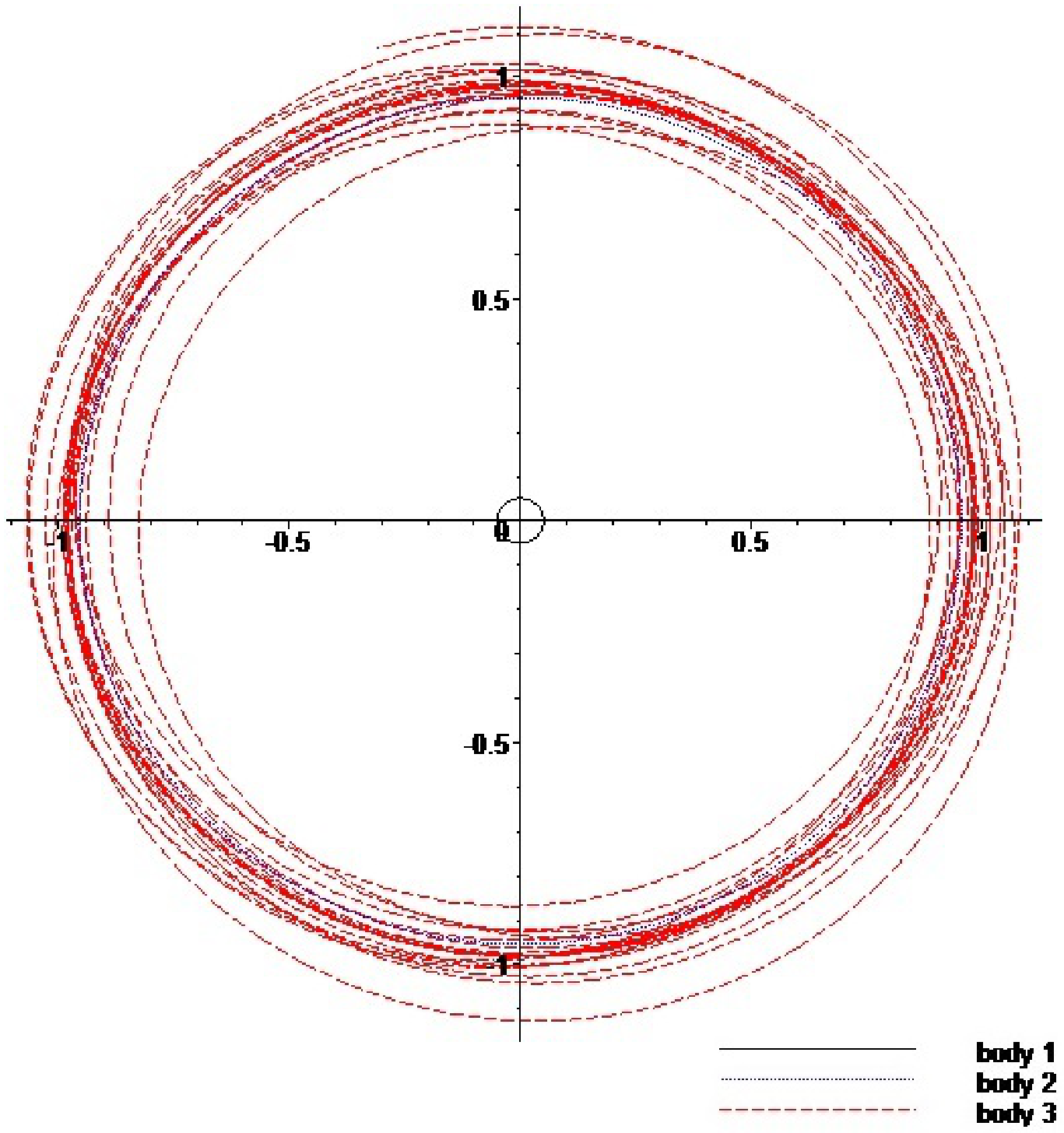}
\caption{DHE, $k=20$, $m=50$, $\mu=0.05$, $\Box=\Box_q$}
\label{figure3a}
\end{minipage} \hfill
\begin{minipage}[c]{.46\linewidth}
\includegraphics[width=4cm]{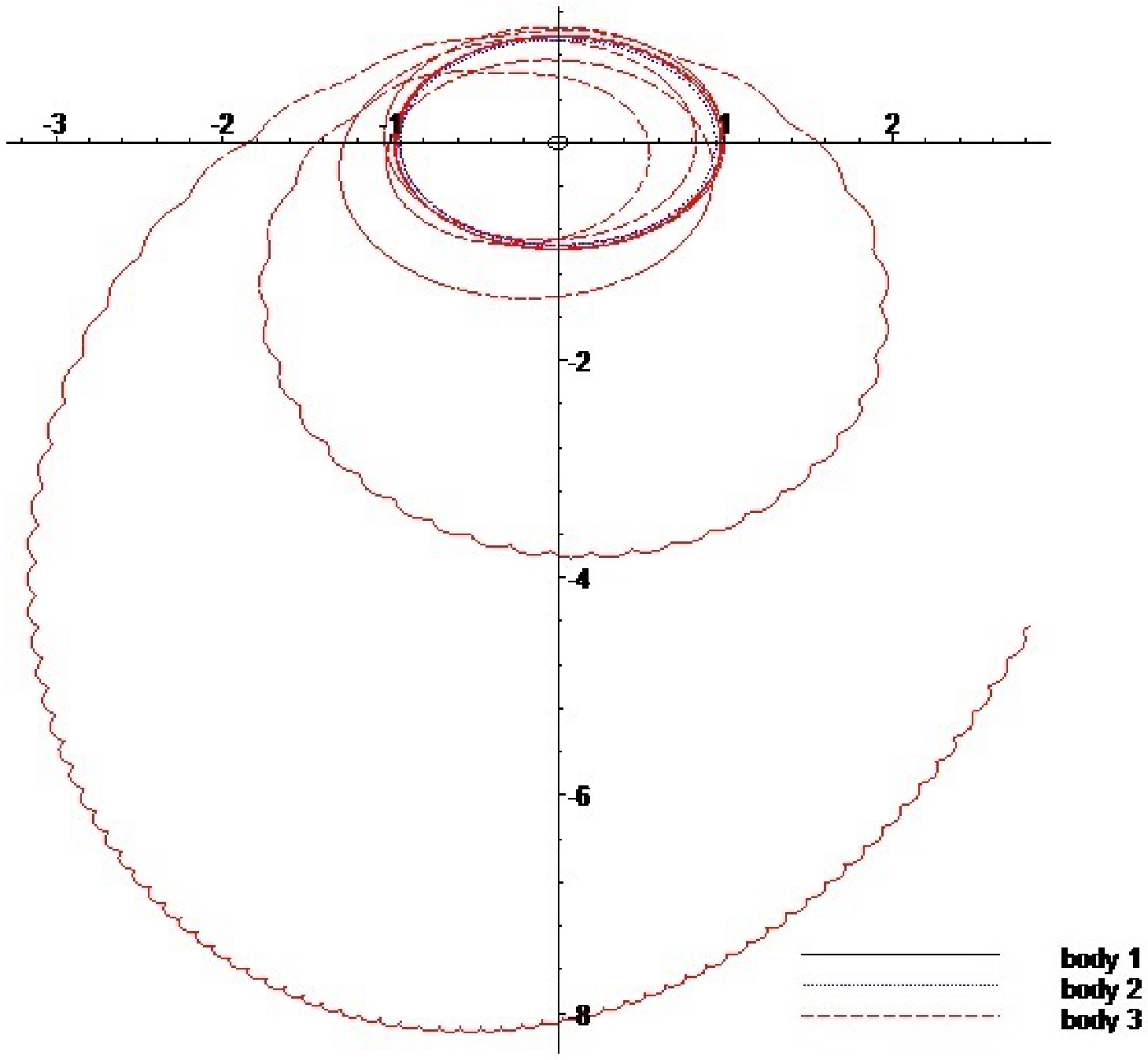}
\caption{DHE, $k=20$, $m=50$, $\mu=0.06$, $\Box=\Box_q$}
\label{figure3b}
\end{minipage}
\end{figure}
As we can see in the two last figures \ref{figure3a} and \ref{figure3b}, as soon as the ratio is greater than the limit value mentioned previously, the system becomes unstable.

\section{Conclusion}

The aim of this work was to apply the formalism of quantum calculus of variations to celestial mechanics.\ As one knows, the search for particular solutions of the many-bodies problem has a particular importance in the historical development of celestial mechanics, probably because of the feeling that toy-models may be realistic and also that the simply-to-state but hard-to-prove questions in this domain are almost all linked with these particular solutions. In that sense, although we did not give the details of its application, the Q.C.V equally applies to other generalized polygonal solutions, see for instance \cite{El1,El2}. However, performing the effective experiments when applying Q.C.V. to choreographic solutions is much more complicated. Let us explain the difficulties in the case of a planar three-bodies choreographic solution such as the remarkable figure-eight solution found in 2000 by A. Chenciner and R. Montgomery \cite{CM}. Keeping the previous notations, and dealing with the grid  $[t_0+2N\varepsilon ,t_f-2N\varepsilon ]\cap (t_\nu +\varepsilon\mathbb{Z})$, the coordinates of three particles $x_{i,n}=x_i(t_0+n\varepsilon )$, $y_{i,n}=y_i(t_0+n\varepsilon )$ with $N\leq n\leq M-N$ and $M\leq (t_f-t_0)/\varepsilon $, $M$ being a multiple of 3, may be expressed as system of algebraic equations with the additional constraints that $a_i,_{n+M/3}=a_{i,n}$, $b_{i,n+M/3}=b_{i,n}$. The main task is to device an efficient method to solve the previous system which cannot be triangularized. We address this issue that is studied in a companion paper of the present one.

The link between the constants of motion of solutions of classical or discrete equations of motion has been established only in the case of relative equilibria. However, for arbitrary solutions, a phenomenon of diffusion of constants of motion appears due to the fact that the classical
derivative and the generalized derivatives do not commute. Lastly, the application of Q.C.V. to systems of particles interacting according to non-homogeneous potentials, for instance those of London and Laplace-Sellinger, is interesting and its treatment may be done through Puiseux series for the solutions of the many-bodies problem in a rotating frame. Indeed in this general situation, we do not have anymore an homothety between the relative equilibria in the newtonian and in the Q.C.V. contexts.




\end{document}